\documentclass[12pt,a4paper]{article}

\pdfoutput=1

\usepackage[
  bookmarks=true,
  bookmarksnumbered=true,
  bookmarksopen=true,
  bookmarksdepth=2,
  pdfborder={0 0 0},
  breaklinks=true,
  colorlinks=true,
  linkcolor=black,
  citecolor=black,
  filecolor=black,
  urlcolor=black,
]{hyperref}
\usepackage[round]{natbib}

\usepackage[margin=1in]{geometry}

\usepackage{amsmath}
\usepackage{amsthm}
\usepackage{amssymb}
\usepackage{mathtools}
\usepackage{isomath}

\usepackage{enumitem}
\usepackage[capitalize]{cleveref}

\usepackage{nicefrac}

\usepackage{fnbreak}
\allowdisplaybreaks

\newcommand{\defnitemtitle}[1]{(#1)\textbf{.}}
\newcommand{\crefpart}[2]{\cref{#1}(\labelcref{#1-#2})}

\newcommand{\Reals}{\mathbb{R}}
\newcommand{\RR}{\Reals_{+}}
\newcommand{\RRn}{\RR^n}
\newcommand{\thresh}{E}
\newcommand{\bound}{H}
\newcommand{\EU}{\mathbb{EU}_{\bound}^n}
\newcommand{\boundTail}{(\bound,\infty)}
\newcommand{\NN}{\mathbb{N}}
\newcommand{\ZZ}{\mathbb{Z}}
\newcommand{\closure}[1]{\overline{#1}}

\DeclareMathOperator{\Max}{Max}
\DeclareMathOperator{\supp}{supp}
\DeclareMathOperator{\poly}{poly}

\newcommand{\mech}{\mathcal{M}}
\newcommand \Rev{\text{\textsc{Rev}}}
\newcommand \SRev{\text{\textsc{SRev}}}
\newcommand{\expect}[2]{\mathbb{E}_{#1}\bigl[#2\bigr]}
\newcommand{\prob}[2]{\mathbb{P}_{#1}\bigl[#2\bigr]}
\newcommand{\probge}[3]{\prob{#1}{#2\!\ge\!#3}}
\newcommand{\probl}[3]{\prob{#1}{#2\!<\!#3}}

\newcommand{\eqdef}{\triangleq}
\newcommand{\epst}{\tilde{\varepsilon}}

\theoremstyle{plain}
\newtheorem{theorem}{Theorem}[section]
\newtheorem{lemma}[theorem]{Lemma}
\newtheorem{sublemma}{Sublemma}[theorem]
\Crefname{sublemma}{Sublemma}{Sublemmas}
\newtheorem{proposition}[theorem]{Proposition}
\newtheorem{corollary}[theorem]{Corollary}
\newtheorem{open problem}[theorem]{Open Problem}

\theoremstyle{definition}
\newtheorem{definition}{Definition}[section]
\newtheorem{remark}{Remark}[section]

\newlist{parts}{enumerate}{1}
\crefname{partsi}{Part}{Parts}
\setlist[parts,1]{label=\alph*.,ref=\alph*}

\hypersetup{
  pdfauthor      = {Moshe Babaioff <moshe@microsoft.com>, Yannai A. Gonczarowski <yannai@gonch.name>, and Noam Nisan <noam.nisan@gmail.com>},
  pdftitle       = {The Menu-Size Complexity of Revenue Approximation},
}

\title{The Menu-Size Complexity of Revenue Approximation\thanks{An extended abstract of this work appeared in STOC 2017 \citep*{bgn17}.}}

\author{Moshe Babaioff\thanks{Microsoft Research, \emph{E-mail}: \href{mailto:moshe@microsoft.com}{moshe@microsoft.com}.}
\and
Yannai A. Gonczarowski\thanks{Microsoft Research, \emph{E-mail}: \href{mailto:yannai@gonch.name}{yannai@gonch.name}. Research conducted while also affiliated with the Einstein Institute of Mathematics, Rachel \& Selim Benin School of Computer Science \& Engineering and Federmann Center for the Study of Rationality, The Hebrew University of Jerusalem, Israel.}
\and
Noam Nisan\thanks{Rachel \& Selim Benin School of Computer Science \& Engineering and Federmann Center for the Study of Rationality, The Hebrew University of Jerusalem, Israel, \emph{E-mail}: \href{mailto:noam.nisan@gmail.com}{noam.nisan@gmail.com}. Research conducted while also affiliated with Microsoft Research.}
}

\date{March 31, 2021}

\begin{document}

\maketitle

\begin{abstract}
Consider a monopolist selling $n$ items to an additive buyer whose item values are drawn from independent distributions $F_1,F_2,\ldots,F_n$ possibly having unbounded support. Unlike in the single-item case, it is well known that the revenue-optimal selling mechanism (a pricing scheme) may be complex, sometimes requiring a continuum of menu entries. Also known is that simple mechanisms with a bounded number of menu entries can extract a constant fraction of the optimal revenue. Nonetheless, whether an arbitrarily high fraction of the optimal revenue can be extracted via a bounded menu size remained open.

We give an affirmative answer: for every $n$ and $\varepsilon>0$, there exists $C=C(n,\varepsilon)$ s.t.\ mechanisms of menu size at most $C$ suffice for obtaining $(1-\varepsilon)$ of the optimal revenue from \emph{any} $F_1,\ldots,F_n$. We prove upper and lower bounds on the revenue-approximation complexity $C(n,\varepsilon)$ and on the deterministic communication complexity required to run a mechanism achieving such an approximation.
\end{abstract}


\section{Introduction}\label{introduction}

As familiar economic institutions move to computerized platforms, they are reaching unprecedented sizes and levels of complexity.
These new levels of complexity
often become the defining feature of the computerized economic scenario, as in the cases of spectrum auctions and
ad auctions.   The use of the word ``complexity'' here is intentionally vague, and can refer to a wide variety of
computational, informational, or descriptive measures of complexity.  A high-level goal of the field of ``Economics and Computation'' is
to analyze such measures of complexity and understand the degree to which they are indeed a bottleneck to achieving desired economic properties.

This paper studies exactly such a question in the recently well-studied scenario of pricing multiple items.  The scenario is that of a monopolist seller who
is selling $n$ items to a single additive buyer.  The buyer has a private value $v_i$ for each item $i$, where each $v_i$ is distributed
according to a commonly known prior distribution $F_i$, independently of the values of the other items.  
The valuation of the buyer is assumed to be additive, so that her value for a subset $S$ of the items
is simply $\sum_{i \in S} v_i$, and the seller's goal is to design an ``auction'' or ``mechanism'' (really just a pricing scheme) that maximizes her revenue.  The
classical economic analysis \citep{Myerson81} shows that for a single item, the optimal mechanism is simply to sell the item at some fixed price.
On the other hand, when there is more than a single item, it is known that the optimal mechanism may be surprisingly
complex, randomized, and unintuitive \citep{MM88,Tha04,MV06,GK14,GK15,HR15,DDT13}.  

A significant amount of recent work has studied whether ``simple'' mechanisms may yield at least approximately optimal revenue.
Following a sequence of results \citep{chawla10b,hart-nisan-a,LY13}, it was shown by \cite{BILW} that one of the following two ``simple'' mechanisms always yields at
least a constant fraction ($\nicefrac{1}{6}$) of the optimal revenue: either sell all items as a single take-it-or-leave-it bundle (for some carefully chosen price)
or sell each item separately for its Myerson price. This was further extended (with different constants) to the case of multiple buyers \citep{Y15} and to
buyers with sub-additive valuations \citep{RW15,CZ17}, but is in contrast to the case where the item values are distributed according to a joint (correlated) distribution,
a case for which no finite approximation is possible by simple mechanisms \citep{pricerand,hart-nisan-b}.  

In this work, we study the trade-off between the complexity of a mechanism and the extent to which it can approximate the optimal revenue.
One may choose various measures of mechanism complexity \citep{hart-nisan-b,DLN14,MR15}, and we will focus on the simplest one, 
the \emph{menu size} suggested in \cite{hart-nisan-b}.
The menu size of a mechanism (for a single buyer, an auction or mechanism is just a pricing scheme) is defined to be the number of different possible outcomes
of the mechanism.  More specifically, every single-buyer mechanism is equivalent to one that offers a \emph{menu} of options to
the buyer, where each option --- \emph{entry} $(\vec{x};p)$ --- in the menu specifies a probability $x_i$ of acquiring each item $i$ as well as a price $p$
to be paid for the combination $\vec{x}$, and the buyer chooses an entry that maximizes her own expected utility $\sum_i x_i \cdot v_i - p$.  The number of
entries in the menu is defined to be the \emph{menu-size complexity} of the mechanism. 
As we observe, the menu size is tightly related
to the {\em deterministic communication complexity} of the mechanism: 
if the mechanism is considered common knowledge, then the amount of information\footnote{See \cref{comm} for a formal definition.} that the
buyer must send to the seller so that the outcome (allocation probabilities and price) of the mechanism for this buyer can be determined by the seller,\footnote{Since only the buyer has private information (her values) in our setting, an arbitrary
interactive protocol between the buyer and the seller must use at least as many bits as this
one-way communication. Thus, this notion captures general two-way deterministic complexity as well.  See \cref{comm} for more details, as well as
a discussion regarding randomized communication complexity.} is exactly equal to the logarithm of the menu size of the mechanism.
It is known that for some distributions,
the optimal mechanism has infinite menu size \citep{DDT13} --- or equivalently by the above, requires infinite communication complexity --- but a constant fraction of the optimal revenue may be 
extracted by a finite-complexity mechanism \citep{BILW}.
Is it possible to extract an arbitrarily high fraction of the optimal revenue via a finite menu size (equivalently, via mechanisms with finite communication complexity)?

Our first and main result shows that, in fact, finite complexity suffices to get \emph{arbitrarily close} to the optimal revenue. 

\begin{definition}[$\Rev_C$; $\Rev$]
For a distribution $F$ on $n$ items, we denote by $\Rev_C(F)$ the 
maximal (formally, the supremum) revenue obtainable by an individually rational incentive-compatible mechanism that has
at most $C$ menu entries
and sells the $n$ items to a single additive buyer whose values for the items are distributed according to $F$.  
We denote by $\Rev(F)=\Rev_{\infty}(F)$
the maximal revenue obtainable without any complexity restrictions on the mechanism.
\end{definition}

Formally, our result shows that $\lim_{C\rightarrow\infty}\frac{\Rev_C(F)}{\Rev(F)}=1$ \emph{uniformly} across all product distributions~$F$. In other words:

\newcounter{maintheorem}
\setcounter{maintheorem}{\value{theorem}}

\begin{theorem}[Qualitative Version]\label{cne-upper-bound}
For every number of items $n$
and every $\varepsilon>0$, there exists a \emph{finite} menu size $C=C(n, \varepsilon)$ such that for every $F_1,F_2,\ldots,F_n\in\Delta(\RR)$, we have that
${\Rev_C(F_1 \times \cdots \times F_n) \ge (1-\varepsilon) \cdot \Rev(F_1 \times \cdots \times F_n)}$.
\end{theorem}

\cref{cne-upper-bound} gives a positive answer to Open Problem 6 from \cite{hart-nisan-combined},\footnote{\cite{hart-nisan-combined} is a manuscript combining \cite{hart-nisan-a} and \cite{hart-nisan-b}.} which asks precisely whether the statement of \cref{cne-upper-bound} holds. From \cref{cne-upper-bound}, we immediately get that finite communication complexity uniformly suffices for arbitrarily approximating the optimal revenue:

\begin{corollary}[Qualitative Version]
For every number of items $n$
and every $\varepsilon>0$, there exists a \emph{finite} communication complexity $D=D(n,\varepsilon)$ such that for every $F_1,F_2,\ldots,F_n\in\Delta(\RR)$, there exists an $n$-item mechanism with a deterministic communication complexity of only $D$ that approximates the optimal revenue from $F_1\times\cdots\times F_n$ up to a multiplicative~$\varepsilon$ loss.
\end{corollary}

It is natural to ask what is the rate of the uniform convergence of the sequence $\frac{\Rev_C(F)}{\Rev(F)}$. In other words, \emph{how complex} must a revenue-approximating mechanism be?

\begin{definition}[Revenue Approximation Complexity]
For every number of items $n$
and every ${\varepsilon>0}$, we define the \emph{revenue approximation complexity} $C(n, \varepsilon)\in\RR$ to be the smallest value~$C\in\RR$ such that 
$\Rev_C(F_1 \times \cdots \times F_n) \ge (1-\varepsilon) \cdot \Rev(F_1\times \cdots \times F_n)$
for every $F_1,F_2,\ldots,F_n\in\Delta(\RR)$.
\end{definition}

The construction used in the proof of \cref{cne-upper-bound} gives an upper bound on $C(n,\varepsilon)$ (i.e., a lower bound on the rate of uniform convergence of $\frac{\Rev_C(F)}{\Rev(F)}$).

\newcounter{theorembackup}
\setcounter{theorembackup}{\value{theorem}}
\setcounter{theorem}{\value{maintheorem}}

\begin{theorem}[Quantitative Version]
$C(n,\varepsilon)\le\bigl(\frac{\log n}{\varepsilon}\bigr)^{O(n)}$.
\end{theorem}

\begin{corollary}[Quantitative Version]
For every number of items $n$,
every $\varepsilon>0$, and every $F_1,F_2,\ldots,F_n\in\Delta(\RR)$, there exists an $n$-item mechanism with a deterministic communication complexity of only
$D(n,\varepsilon)=\log C(n,\varepsilon)=O\bigl(n \log\bigl(\frac{\log n}{\varepsilon}\bigr)\bigr)$ that approximates the optimal revenue from $F_1\times\cdots\times F_n$ up to a multiplicative $\varepsilon$ loss.
\end{corollary}

\setcounter{theorem}{\value{theorembackup}}

This bound on the menu size is exponential in $n$, and so the next natural question is whether polynomial menu size suffices.  
At first glance the answer seems to be
``obviously not'': the menu-size complexity measure is quite weak, and even the mechanism that sells each item separately has exponential
menu size (since, when presented as a menu, a menu entry is needed for each possible subset of the items), which has been the source of one of the main criticisms of the menu-size complexity measures.  This answer, however, is premature; in fact, we show that polynomial
menu size turns out to suffice for approximating the revenue obtainable from selling items separately, thereby also appeasing this criticism.\footnote{That is, by our next result, any approximation impossibility that can be shown for polynomial menu-size mechanisms would immediately apply also to separate selling as well.}  Let us denote by $\SRev(F_1 \times \cdots \times F_n)$ the revenue obtainable
by selling each item separately for its optimal price.

\begin{theorem}\label{poly-approx-srev}
For every $\varepsilon>0$, there exists $d(\varepsilon)$ such that for every number of items~$n$
and $F_1,F_2,\ldots,F_n\in\Delta(\RR)$, we have for $C=n^{d(\varepsilon)}$ that
$\Rev_C(F_1 \times \cdots \times F_n) \ge (1-\varepsilon) \cdot \SRev(F_1 \times \cdots \times F_n)$. 
\end{theorem}

The same bound applies also to the revenue obtainable by selling the items after arbitrarily prepartitioning them into bundles.
Using the result of \cite{BILW}, this immediately implies that polynomial menu size suffices for extracting a constant fraction of the \emph{optimal} revenue.

\begin{corollary}\label{poly-approx-rev}
There exist a fixed constant number $d$ and a fixed constant fraction $\alpha>0$ such that $C(n, 1-\alpha) \le O(n^d)$.
\end{corollary}
\noindent The above reasoning shows that \cref{poly-approx-rev} holds for every $\alpha$ arbitrarily close to $\nicefrac{1}{6}$, which is the constant fraction of the optimal revenue shown by \cite{BILW} to be obtainable by the better of bundled selling and separate selling.

Does polynomial menu size suffice for extracting revenue \emph{arbitrarily close} to the optimal revenue?  We prove that
this is not the case, at least for $\varepsilon$ that is polynomially small in $n$.

\begin{theorem}\label{cne-lower-bound}
$C(n, \nicefrac{1}{n}) \ge 2^{\Omega(n)}$.
\end{theorem}
\noindent The proof of \cref{cne-lower-bound} shows, in fact, that polynomial dependence on
$\varepsilon$ is impossible even for approximating the revenue from selling the items separately.

At this point, we leave two main problems open.  The first one is whether for every \emph{fixed} $\varepsilon>0$, polynomial
(or at least quasi-polynomial)
menu size suffices for approximating the optimal revenue to within a multiplicative $\varepsilon$.   
In terms of communication complexity, this translates to whether
logarithmic or polylogarithmic deterministic communication suffices\footnote{The authors have opposing conjectures regarding the answer to this open problem.} for every fixed value of $\varepsilon$.

\begin{open problem}\label{arbitrary-eps}
Is it true that for every $\varepsilon>0$, there exists $d(\varepsilon)$ such that $C(n, \varepsilon) \le O\bigl(n^{d(\varepsilon)}\bigr)$? 
How about $C(n, \varepsilon) \le O\bigl(2^{\log^{d(\varepsilon)} n}\bigr)$? 
\end{open problem}

The second open problem (or rather, class of open problems) is whether stronger notions of mechanism description complexity 
may allow for better revenue in 
polynomial complexity.   This
may be asked with respect to any complexity measure, and it is not clear which specific natural choice to consider, so 
an identification of such a measure is part of what is left open.\footnote{One may be tempted to consider 
the \emph{additive menu size} complexity defined by \cite{hart-nisan-b}, which
allows the seller to present menus from which the buyer is allowed to take \emph{any combination} of menu entries for the sum of their prices. (So, for example, selling each item separately has linear additive menu size.)
There are several possible interpretations here regarding whether buying lotteries translates into winning at the sum of probabilities (capped at one) or translates into not winning at the product of the probabilities of not winning. Similarly, the mechanism could either allow or not allow the buyer to adaptively decide whether and which additional menu entry to buy based on the realization of the lottery in the first, already purchased, menu entry. Follow-up work by \cite{bnr18} showed that regardless of which interpretation is adopted, additive menus \emph{cannot} achieve the optimal revenue. In other words, restricting to additive menus in fact entails for some distributions a constant multiplicative loss on the revenue, and so additive menus (regardless of how large their menu size is) cannot be used to obtain arbitrarily small losses compared to the optimal revenue.
There are several more general definitions that are possible, but we have not found a truly satisfactory one.}
 
\subsection{Further and Follow-Up Literature}

In the short time since the first appearance of this paper, there has been a rise in the research attention dedicated to menu sizes, in two major directions. Most related to this work, \cite{g18} analyzes the required menu size when holding the number of items $n$ fixed, as a function of the allowed loss $\varepsilon$ (in a sense symmetrically to our \cref{arbitrary-eps} above, which asks about the menu size for loss $\varepsilon$ held fixed, as a function of the number of items $n$), and shows that already for two items with bounded i.i.d.\ valuations, a polynomial dependence on $\nicefrac{1}{\varepsilon}$ as in our bound from \cref{cne-upper-bound} is required. As \cite{g18} shows, this lower bound, together with our upper bound (\cref{cne-upper-bound}) and with our observation that the deterministic communication complexity of running a mechanism is the logarithm of its menu size (\cref{cc} in \cref{comm}), implies the following: for any fixed number of items~$n$, there exists a tight $\Theta(\log\nicefrac{1}{\varepsilon})$ bound on the minimal communication complexity that is guaranteed for any product distribution to suffice for running a mechanism that maximizes revenue up to a multiplicative~$\varepsilon$.

Our arguments readily extend beyond additive valuations, to subadditive valuations. Following up on our work, \cite{KMSSW19} consider a more permissive notion of menu size that they define, \emph{symmetric menu size}, which allows for a single menu entry to represent many menu entries that can be obtained from the original menu entry by permuting the identities of the items. They extend and tighten our analysis for the case of unit-demand valuations (a different subclass of subadditive valuations) and for their menu-size notion (their most notable change is to the last part of our analysis --- the discretization of the ``cheap part'' of the menu), and prove that a quasipolynomial symmetric menu size suffices for obtaining a $(1-\varepsilon)$ fraction of the optimal revenue. That is, they give a positive answer to a unit-demand/symmetric-menu-size variant of the second part of our \cref{arbitrary-eps}. They still explicitly leave open the question of obtaining a similar result for additive buyers (as in our \cref{arbitrary-eps}), even in their more permissive symmetric menu size model.

Slightly farther apart is the line of study of the menu size of mechanisms that lie, in a sense, between single- and multi-dimensional mechanisms --- mechanisms whose study was originated by \cite{fgkk16} with their study of the ``FedEx auction.'' Necessary and sufficient menu sizes for precise revenue maximization in the FedEx setting and variants thereof have been studied by \cite{fgkk16}, \cite{dw17}, \cite{ssw18}, and \cite{dgssw18}, while the menu size for approximate revenue maximization up to a multiplicative loss of $\varepsilon=\nicefrac{1}{n^2}$ has been studied by \cite{ssw18}. This latter analysis asks a somewhat similar question to our \cref{cne-lower-bound} (where~$\varepsilon$ also shrinks as $n$ grows, though with $\varepsilon=\nicefrac{1}{n}$), however in the multi-item model that we study the corresponding menu size is exponential while in the ``one-and-a-half dimensional'' FedEx model that they study it is polynomial. For more details on all of these results, we refer the interested reader to \cite{gg18-tutorial}.


\section{Upper Bound on\texorpdfstring{\\}{ }Revenue-Approximation Complexity}\label{upper-bound}

In this \lcnamecref{upper-bound}, we prove our main result, \cref{cne-upper-bound}, which states that $C(n,\varepsilon)$ is finite for every number of items~$n\in\NN$ and $\varepsilon>0$, and moreover, that~$C(n,\varepsilon)\le\bigl(\frac{\log n}{\varepsilon}\bigr)^{O(n)}$. To somewhat ease presentation and to emphasize which elements in the proof are used for the qualitative result and which only for the precise quantitative one, we will first prove that $C(n,\varepsilon)\le(\nicefrac{n}{\varepsilon})^{O(n)}$, which makes the same qualitative statement and still requires virtually all of the novel technical ``beef'' of the proof, but saves some clutter at the end. The proof proceeds in four steps. \cref{upper-bound-overview-sec} provides a rough overview of the proof strategy, \cref{no-two-high-sec,expensive-exc-sec,expensive-trim-sec,cheap-disc-sec} provide the formal details of each of the four steps of the proof, and \cref{connect-dots-sec} connects the dots by combining the four steps. \cref{tighter} gives an overview of the modifications required to prove the quantitatively stronger upper bound of $\bigl(\frac{\log n}{\varepsilon}\bigr)^{O(n)}$, with the details relegated to the \lcnamecref{logn}. Finally, \cref{exclusively-unbounded-sec} concludes with a short discussion of the application of the proof steps to obtain uniform approximation results for correlated distributions over a restricted valuation space, which generalize bounded distributions.

\subsection{Proof Overview}\label{upper-bound-overview-sec}

Let $F_1, F_2,\ldots,F_n$ be the respective distributions of the values of the $n$ items.
We will construct a mechanism with (finite) menu size $(\nicefrac{n}{\varepsilon})^{O(n)}$ that guarantees a $(1-\varepsilon)$ multiplicative approximation to the optimal revenue.\footnote{The overview of the proof of the stronger upper bound of $(\frac{\log n}{\varepsilon})^{O(n)}$ is identical up to the last step, which is also quite similar but more intricate for the stronger bound; See \cref{tighter} for the details.}

\paragraph{Limitations of Existing Techniques} One possible approach to approximate revenue maximization, taken by \cite{LY13}, \cite{BILW}, \cite{RW15}, and more recently \cite{CZ17}, is to use a core/tail decomposition and bound the revenue from (or the welfare of) the core and the revenue from the tail. Unfortunately, such a decomposition inherently entails a nonnegligible revenue loss (it only guarantees a constant fraction of the optimal revenue), either due to bounding the welfare of the core instead of the revenue from it, or due to estimating the total revenue using the revenues obtained by selling to the core and to the tail separately. Therefore, while this approach makes no assumptions regarding the valuation space (beyond independence), this technique, as used in the literature so far, is unsuitable for guaranteeing negligible loss in revenue as in the result that we seek.

Another possible approach, taken by \cite{DW12}, \cite{hart-nisan-b}, and \cite{DLN14}, is to round all possible menu entries onto a discrete grid via ``nudge and round'' operations.
Unfortunately, for the grid (and thus the menu size) to be finite and for the revenue loss to indeed be negligible, the above papers all require that the valuation space be bounded. Therefore, while this approach can guarantee negligible loss in revenue, this technique, as used in the literature so far, is unsuitable for the analysis of unbounded valuation spaces as in our setting.

\medskip

To overcome the above-described limitations of the core/tail decomposition technique, we take a more subtle approach, by analyzing core and tail regions together. We first show (in Step 1 below) that one does not lose much revenue by disregarding what can be described as ``second-order'' tails, i.e., valuations where two or more of the item values lie in the tail. Then, we show how to gradually simplify an optimal mechanism (which may be arbitrarily complex, even infinite in size) for the valuation space consisting of the core plus all first-order tails while losing only a tiny fraction of the revenue in each step. For every modification that we perform to the menu, we must ``simultaneously'' check that we do not significantly hurt the revenue from either core or (first-order) tail buyers. To gradually simplify the menu, we first carefully modify the menu so that only a small number of menu entries have a high price (this is the most technically elaborate part of the proof, performed in Steps 2 and 3 below), and then (in Step 4 below) round the menu entries with low prices to a finite grid using ``nudge and round'' operations. At this point, the use of ``nudge and round'' onto a finite grid is possible without significant revenue loss since the price of the menu entries that we round is bounded. Nonetheless, care still has to be taken beyond previous ``nudge and round'' uses, to ensure that this rounding does not incentivize buyers in the (first-order) tails to switch to buying a lower-priced rounded entry. Before moving on to the definitions and formal statements and proof, we first give a somewhat more detailed, yet still high-level, overview of each of the four steps of the proof.

\paragraph{Step 1: Move to an ``almost bounded'' valuation space} This step, taken in \cref{no-two-high-sec}, simplifies the valuation space by showing that since item values are independent, finding an approximately optimal mechanism under the assumption that at most one item has a value higher than~$\bound$ (i.e., has a value that lies in the $\bound$-tail), for some $\bound=\poly(n,\nicefrac{1}{\varepsilon})$, entails a very small loss compared to doing so without this assumption. This is possible, very roughly speaking, because the probability of two item values lying in the tail, for $H$ as above, can be thought of as being of order $\varepsilon^2$, while the revenue conditioned upon being in this ``second-order'' tail (i.e., conditioned upon the values of both of these items lying in the $\bound$-tail) is of order~$\nicefrac{1}{\varepsilon}$. Therefore, it is enough to construct our finite approximation for the distribution conditioned upon being in the valuation space comprised of the core and the first-order tail, i.e., the valuation space where at most one item value lies in the tail. We call distributions over this valuation space \emph{exclusively unbounded} distributions. We note that this is the only step in which the independence of the item values is used; indeed, combining the remaining steps shows that the revenue from all exclusive unbounded distributions (even highly correlated distributions not originating from a product distribution over~$\RRn$) can be uniformly approximated using finite-size menus (see \cref{upper-bound-correlated} in \cref{exclusively-unbounded-sec}).

\paragraph{Step 2: Modify expensive menu items to behave ``almost like'' single-item mechanisms} This step, taken in \cref{expensive-exc-sec}, starts with an optimal (possibly arbitrarily complex) revenue-maximizing mechanism for some exclusively unbounded distribution. In this step, we simplify the ``expensive'' part of the menu, i.e., the part of the menu consisting of all menu entries that cost more than $\thresh$, for some $\thresh=\poly(n,\nicefrac{1}{\varepsilon})$, so that each expensive menu entry allocates only a single item with non-zero probability.  This means that while in the ``cheap'' part of the menu we can allocate arbitrary combinations of items, once the price increases beyond $\thresh$, our mechanism must act like a unit-demand one and never allocate more than a single item. We call such a mechanism \emph{$\thresh$-exclusive}. This is possible since, roughly speaking, due to the assumption of exclusive unboundedness, most of the value from an expensive menu entry chosen by some buyer type comes only from the unique item whose value lies in the tail for the valuation of that buyer type. Thus, instead of offering that (nonexclusive) menu entry, we offer an (exclusive) entry with only the corresponding winning probability of that item, for a slightly discounted price. While in most natural cases, this step in fact increases the size of the expensive part of the menu (as each expensive menu entry possibly becomes $n$ exclusive menu entries, each allocating a distinct item with non-zero probability), this simplification allows the next step to significantly reduce the size of this part of the menu.

\paragraph{Step 3: Apply \citeauthor{Myerson81}'s result to obtain ``almost one'' expensive entry per item} This step, taken in \cref{expensive-trim-sec}, reduces the size of the expensive part of the menu to at most $2n$ menu entries. This is the most technically elaborate step. Since $\thresh$-exclusivity means that the expensive menu entries ``look like'' separate single-dimensional mechanisms for each of the $n$ items, we show that we are able to carefully use the analysis of \cite{Myerson81} to replace each of these separate expensive mechanisms with a simple ``almost deterministic'' one.  In contrast to \citeauthor{Myerson81}'s single (non-zero) menu entry, we require two menu entries for each item: a deterministic one analogous to \citeauthor{Myerson81}'s ``optimal price'' entry, and an additional randomized one analogous to the ``opt out'' zero entry in \citeauthor{Myerson81}'s mechanism. The function of the latter entry is to make sure that buyers are not incentivized to ``jump'' from the expensive part to the cheap part of the menu following the reduction of the size of the former.

\paragraph{Step 4: Discretize cheap menu entries ``almost to a grid''} This final step, taken in \cref{cheap-disc-sec}, simplifies the cheap part of the menu by ``rounding'' the menu entries into a discrete set.
We note that even at this point in the proof, the ``nudge and round'' techniques that allowed this rounding to be done with only negligible loss of revenue for bounded valuations in previous papers \citep{DW12,hart-nisan-b,DLN14} cannot just be used ``out of the box'' in this step. Indeed, slight changes in allocation probabilities may result in large revenue changes, since the valuation space is not bounded but only exclusively unbounded. Nonetheless, these techniques can be carefully extended to be used here as well. Roughly speaking, we construct $n$ discretizations of each cheap menu entry, where each discretization rounds the price and all but one allocation; rounding in the right direction guarantees that at least one of these discretizations is still a leading candidate for any buyer type that previously chose the corresponding original (nondiscretized) menu entry. As all but one coordinate of each of the discretized menu entries lie on a grid, we show that only finitely many of the menu entries are in fact chosen by any buyer type.

\medskip

While the second and third (and first) steps each entail a slight multiplicative revenue drop, the fourth step entails also a slight additive revenue drop.
Recall, however, that we aim to achieve only a slight multiplicative drop (with no additional additive drop) in overall revenue.
To obtain this result, when combining all of the above steps in \cref{connect-dots-sec} we assume w.l.o.g.\ that $\Max_i\Rev(F_i)$ is normalized\footnote{The cases in which $\Max_i\Rev(F_i)$ cannot be normalized, i.e., when it is $0$ or infinite, are easy to handle separately. In the former case, there is nothing to show. In the latter case, $\Rev(F_i)=\infty$ for some $i\in[n]$, and so by the theorem of \cite{Myerson81}, an arbitrarily high revenue can be extracted using a take-it-or-leave-it offer for item $i$.} (by scaling the currency) to a suitable value such that the additive drop in the fourth step can be quantified to be less than a slight multiplicative drop. Clearly, as the obtained bound on the overall cumulative revenue drop for normalized mechanisms is purely multiplicative, the proof also implies the same multiplicative bound for all (even nonnormalized) distributions.

\subsection{Preliminaries}

\begin{definition}[Notation]\leavevmode
\begin{itemize}
\item
\defnitemtitle{Naturals}
We denote the strictly positive natural numbers by $\mathbb{N}\eqdef\{1,2,3,\ldots\}$.
\item
\defnitemtitle{[n]}
For every $n \in \mathbb{N}$, we define $[n]\eqdef\{1,2,\ldots,n\}$.
\item
\defnitemtitle{Nonnegative Reals}
We denote the nonnegative reals by $\RR\eqdef\{r \in \Reals \mid r \ge 0\}$.
\item
\defnitemtitle{$\Delta(\cdot)$}
For a set $A$, we denote by $\Delta(A)$ the set of probability distributions over $A$.
\end{itemize}
\end{definition}

\begin{definition}[Outcome; Type; Utility]
Let $n\in\mathbb{N}$ be a number of items. 
\begin{parts}
\item
An \emph{outcome} is an $(n+1)$-tuple $(\vec{x};p)=(x_1,x_2,\ldots,x_n;p)\in[0,1]^n\times\RR$, denoting an allocation (to the buyer) of every item $i\in[n]$ with probability $x_i$, for a total price (paid by the buyer) of $p$.
\item
We denote the (expected) \emph{utility} of a (risk-neutral additive) buyer with \emph{type} (respective item valuations) $v=(v_1,v_2,\ldots,v_n)\in\RRn$ from an outcome $e=(\vec{x};p)\in[0,1]^n\times\RR$ by \[u_e(v)\eqdef \sum_{i=1}^n x_i\cdot v_i - p.\]
\end{parts}
\end{definition}

\begin{definition}[IC Mechanism as Menu]
Let $n\in\mathbb{N}$ be a number of items. By the taxation principle, we identify any \emph{incentive-compatible (IC)} $n$-item \emph{mechanism} with a (possibly infinitely large) menu of outcomes (the \emph{entries} in the menu are all the possible outcomes of the mechanism), where by IC the buyer chooses an entry that maximizes her utility.\footnote{If the menu is infinite, then the fact that it corresponds to an IC mechanism guarantees that some menu entry maximizes the utility of each buyer type. See \cref{closure-max-utility} for more details.} If the mechanism is \emph{individually rational (IR)}, then we assume w.l.o.g.\ that the menu includes the entry $(\vec{0};0)$ that allocates no item and costs nothing. (Conversely, if the menu includes the entry $(\vec{0};0)$, then the mechanism is IR.) Following \cite{hart-nisan-b}, we define the \emph{menu size}  of an IC and IR mechanism as the number of entries, except $(\vec{0};0)$, in the menu of that mechanism.
\end{definition}

\begin{definition}[$\Rev_{\mech}$; $\Rev_C$; $\Rev$]\label{revenue}
Let $n\in\mathbb{N}$ be a number of items and let $F\in\Delta(\RRn)$ be a distribution over~$\RRn$.
\begin{parts}
\item\label{revenue-mech}
Given an IC and IR $n$-item mechanism $\mech$, we denote the (expected) revenue obtainable by $\mech$ from (a single risk-neutral additive buyer with type distributed according to) $F$, by \[
\Rev_{\mech}(F)\eqdef\expect{v\sim F}{p(v)},\]
where $p(v)$ is the price of the entry from $\mech$ that maximizes the utility of $v$, with ties broken in favor of higher prices.\footnote{The results of this paper hold regardless of the tie-breaking rule chosen. See \cref{tie-breaking} for more details.}\textsuperscript{,}\footnote{If the menu is infinite, then the fact that a utility-maximizing menu entry exists for every buyer type does not guarantee that a utility-maximizing entry \emph{with maximal price} (among all utility-maximizing entries) exists for every buyer type. (I.e., it is not guaranteed that the supremum price over all utility-maximizing entries is attained as a maximum.) Indeed, to be completely general, a more subtle definition of the revenue obtainable by an IC mechanism would have been needed. Nonetheless, for the mechanisms considered in this paper, this subtle definition is not required as we make sure that they all possess, for each buyer type, a utility-maximizing entry with maximal price. See \cref{closure-max-price} for more details.}
\item
Given $C\in\NN$, we denote the highest revenue (more accurately, the supremum of the revenues) obtainable from $F$ by an IC and IR $n$-item mechanism with at most $C$ menu entries by \[\Rev_C(F)\eqdef\smashoperator[r]{\sup_{\substack{\mech\subseteq[0,1]^n\times\RR:\\|\mech|\le C}}}\Rev_{\mech}(F).\]
\item
We denote the highest revenue (more accurately, the supremum of the revenues) obtainable from $F$ by an IC and IR $n$-item mechanism by \[\Rev(F)\eqdef\smashoperator[r]{\sup_{\mech\subseteq[0,1]^n\times\RR}}\Rev_{\mech}(F).\]
\end{parts}
\end{definition}

\begin{theorem}[\citealp{hart-nisan-a}]\label{sum-rev}
	$\Rev(F \times G) \le 2\cdot\bigl(\Rev(F)+\Rev(G)\bigr)$,
	for every $m,n\in\NN$, $F\in\Delta(\RR^m)$, and $G\in\Delta(\RRn)$.
\end{theorem}

\subsection{At Most One High Value}\label{no-two-high-sec}

As outlined above, our first step toward proving \cref{cne-upper-bound}, which we take in this \lcnamecref{no-two-high-sec}, simplifies the valuation space. It does so by showing that since item values are independent, any mechanism that extracts most of the revenue under the assumption that the valuation space is restricted to some \emph{$\bound$-exclusively unbounded} valuation space, i.e., to a valuation space where for each buyer type at most one item has value higher than some $\bound=\poly(n,\nicefrac{1}{\varepsilon})$, also extracts most of the revenue without this assumption. This subsection is dedicated to the statement and proof of \cref{no-two-high}, which formalizes this step.

\begin{sloppypar}
\begin{definition}[$\EU$; Exclusively Unbounded Type Distribution]
Let $\bound\in\RR$ and $n\in\NN$.
\begin{parts}
\item
We denote the subset of $\RR^n$ where at most one coordinate is strictly greater than $\bound$ by
\[
\EU \eqdef \bigl\{(v_1,v_2,\ldots,v_n)\in\RRn ~\big|~ |\{i\in[n] \mid v_i > \bound\}| \le 1 \bigr\}.
\]
\item
We say that a type distribution $F\in\Delta(\RRn)$ is \emph{($\bound$-)exclusively unbounded} if $\supp(F)\subseteq\EU$.
\end{parts}
\end{definition}
\end{sloppypar}

\begin{definition}[$F|_{A}$]
For a set $A$ and a distribution $F$ defined over some superset of $A$ s.t.\ $A$ is measurable and $F(A)>0$, we denote the conditional distribution of $v\sim F$ conditioned upon $v\in A$ by $F|_{A}$. Formally, for every measurable set $B\subseteq A$, we define $F|_{A}(B) \eqdef \frac{F(B)}{F(A)}$.
\end{definition}

\begin{lemma}\label{no-two-high}
Let $n\in\NN$ s.t.\ $n\ge2$, let $R\in\RR$, let $\varepsilon\in(0,1)$, and let $\bound\ge\frac{2 \cdot n \cdot (n-1) \cdot R}{\varepsilon}$.  For every $F = F_1\times F_2 \times \cdots \times F_n \in \Delta(\RR)^n$ s.t.\ $\Max_{i\in[n]}\Rev(F_i)\le R$, all of the following hold.
\begin{parts}
\item\label{no-two-high-well-defined}
$F(\EU)>0$ (hence the exclusively unbounded conditioned distribution $F|_{\EU}$ is well defined).
\item\label{no-two-high-rev}
$\Rev(F|_{\EU})\ge(1-\varepsilon)\cdot\Rev(F)$.
\item\label{no-two-high-revm}
For every $a\in(0,1]$ and for every IC and IR $n$-item mechanism $\mech$, if $\Rev_{\mech}(F|_{\EU}) \ge {a\cdot\Rev(F|_{\EU})}$, then $\Rev_{\mech}(F) \ge (1-\varepsilon)\cdot a \cdot \Rev(F)$.
\end{parts}
\end{lemma}

\begin{proof}
For every $i\in[n]$, we denote the probability of $v\sim F_i$ being greater than $\bound$ by $p_i \eqdef F_i\bigl(\boundTail\bigr)$.
We first note that
\begin{equation}\label{no-two-high-p}
p_i\le\frac{\varepsilon}{2\cdot n \cdot (n-1)}
\end{equation} for every $i\in[n]$. Indeed, the revenue from $F_i$ of the mechanism selling (item~$i$) for a take-it-or-leave-it price of $\bound$ is at least $\bound\cdot p_i$, and by definition of $\Rev$ we therefore have $\bound\cdot p_i \le \Rev(F_i) \le R$ and so $p_i \le \nicefrac{R}{\bound} \le \frac{\varepsilon}{2\cdot n \cdot (n-1)}$, as claimed. In particular, since $\varepsilon<1$, we note that this implies that $F_i\bigl([0,\bound]\bigr)=1-p_i>0$ and so, since $\EU\supset[0,\bound]^n$, we obtain that $F(\EU)\ge\prod_{i=1}^n(1-p_i)>0$, proving \cref{no-two-high-well-defined}. (Thus, $F|_{\EU}$ is well defined.)

For the proof of \cref{no-two-high-rev,no-two-high-revm}, we will need the following \lcnamecref{cover-rev} (the second part of this \lcnamecref{cover-rev} is a slightly generalized version of the ``subdomain stitching'' lemma of \cite{BILW}; we give a full proof below for completeness).

\begin{sublemma}\label{cover-rev}
Let $n\in\NN$ and let $F\in\Delta(\RRn)$.
\begin{parts}
\item\label{cover-rev-sub}
$F(B)\cdot \Rev(F|_B) \le F(A) \cdot \Rev(F|_A)$ for every\footnote{If $F(C)=0$ for some $C\subseteq\RRn$, then even though $F|_C$ is not defined, we henceforth define $F(C)\cdot\Rev(F|_C)$ to equal~$0$.} $B\subseteq A \subseteq \RRn$.
\item\label{cover-rev-cover}
For every $m\in\NN$ and every $A_1,A_2,\ldots,A_m\subseteq\RRn$ s.t.\ $\bigcup_{i=1}^m A_i = \RRn$,
we have that $\Rev(F) \le \sum_{i=1}^{m} F(A_i) \cdot \Rev(F|_{A_i})$.
\end{parts}
\end{sublemma}

The proof of \cref{cover-rev} is given after the proof of \cref{no-two-high}.
We now proceed to prove \cref{no-two-high-rev,no-two-high-revm} of \cref{no-two-high}.
We note that for every $i\in[n]$, if $p_i>0$, then
\begin{equation}\label{no-two-high-rev-tail}
\Rev(F_i|_{\boundTail}) \le \nicefrac{\Rev(F_i)}{p_i}.
\end{equation}
Indeed, by \crefpart{cover-rev}{sub}, we have that $p_i \cdot \Rev(F_i|_{\boundTail}) \le \Rev(F_i)$, and so $\Rev(F_i|_{\boundTail}) \le \nicefrac{\Rev(F_i)}{p_i}$, as required.

For every $1\le i<j\le n$, we let
$B_{i,j}\eqdef\bigl\{(v_1,v_2,\ldots,v_n)\in\RRn~\big|~v_i>\bound\And v_j>\bound\bigr\}$ (the ``double-tail'' w.r.t.\ $i$ and $j$) and $p_{i,j}\eqdef F(B_{i,j})=p_i\cdot p_j$.
We claim that
\begin{equation}\label{no-two-high-bij}
p_{i,j}\cdot\Rev(F|_{B_{i,j}}) \le \frac{\varepsilon}{\binom{n}{2}}\cdot\Rev(F),
\end{equation}
for every $1\le i<j\le n$. Since the claim trivially holds when $p_{i,j}=0$, we need only prove it when $p_i>0$ and $p_j>0$. In this case, by \cref{sum-rev} (applied twice) and by \cref{no-two-high-rev-tail,no-two-high-p}, we have
\begin{align*}
p_{i,j}\cdot\Rev(F|_{B_{i,j}}) & =
p_{i,j}\cdot\Rev\Bigl(F_i|_{\boundTail} \times F_j|_{\boundTail} \times \smashoperator[r]{\bigtimes_{k\in[n]\setminus\{i,j\}}}F_k\Bigr) \le \\
& \le 4\cdot p_{i,j}\cdot\Bigl(\Rev\bigl(F_i|_{\boundTail}\bigr)+\Rev\bigl(F_j|_{\boundTail}\bigr)+\Rev\bigl(\smashoperator[r]{\bigtimes_{k\in[n]\setminus\{i,j\}}}F_k\bigr)\Bigr) \le \\
& \le 4\cdot p_{i,j}\cdot\Bigl(\tfrac{\Rev(F_i)}{p_i}+\tfrac{\Rev(F_j)}{p_j}+\Rev\bigl(\smashoperator[r]{\bigtimes_{k\in[n]\setminus\{i,j\}}}F_k\bigr)\Bigr)= \\
& = 4\cdot\Bigl(p_j\cdot\Rev\bigl(F_i\bigr)+p_i\cdot\Rev\bigl(F_j\bigr)+p_i\cdot p_j\cdot\Rev\bigl(\smashoperator[r]{\bigtimes_{k\in[n]\setminus\{i,j\}}}F_k\bigr)\Bigr) \le \\
& \le 4\cdot \frac{\varepsilon}{2 \cdot n \cdot (n-1)}\Bigl(\Rev\bigl(F_i\bigr)+\Rev\bigl(F_j\bigr)+\Rev\bigl(\smashoperator[r]{\bigtimes_{k\in[n]\setminus\{i,j\}}}F_k\bigr)\Bigr) \le \\
& \le 2\cdot\frac{\varepsilon}{n\cdot (n-1)} \cdot \Rev(F) = \frac{\varepsilon}{\binom{n}{2}}\cdot\Rev(F),
\end{align*}
as claimed.

We define $p_{\mathcal{EU}}\eqdef F(\EU)$.
As $\RRn\setminus\EU=\bigcup_{1\le i<j\le n} B_{i,j}$, by \crefpart{cover-rev}{cover} we have that $\Rev(F) \le p_\mathcal{EU}\cdot\Rev(F|_{\EU}) + \sum_{1\le i<j\le n} p_{i,j}\cdot\Rev(F|_{B_{i,j}})$. Therefore, by \cref{no-two-high-bij}, we have that
\begin{multline}\label{drop-b}
p_\mathcal{EU}\cdot\Rev(F|_{\EU}) \ge \Rev(F) - \smashoperator{\sum_{1\le i<j\le n}} p_{i,j}\cdot\Rev(F|_{B_{i,j}}) \ge \\
\ge \Rev(F) - \binom{n}{2}\cdot\frac{\varepsilon}{\binom{n}{2}}\cdot\Rev(F) = (1-\varepsilon)\cdot\Rev(F).
\end{multline}
In particular, $\Rev(F|_{\EU})\ge p_\mathcal{EU}\cdot\Rev(F|_{\EU}) \ge (1-\varepsilon)\cdot\Rev(F)$, proving \cref{no-two-high-rev}.

Let $a\in(0,1]$ and let $\mech$ be an IC and IR $n$-item mechanism with $\Rev_{\mech}(F|_{\EU})\ge a\cdot\Rev(F|_{\EU})$. By definition of $\Rev$ and by \cref{drop-b}, we have that
\begin{multline*}
\Rev_{\mech}(F) = p_\mathcal{EU} \cdot \Rev_{\mech}(F|_{\EU}) + (1-p_\mathcal{EU}) \cdot \Rev_{\mech}(F|_{\RRn\setminus\EU}) \ge \\
\ge p_\mathcal{EU} \cdot \Rev_{\mech}(F|_{\EU}) \ge
p_\mathcal{EU} \cdot a\cdot\Rev(F|_{\EU}) \ge
a\cdot(1-\varepsilon)\cdot\Rev(F),
\end{multline*}
proving \cref{no-two-high-revm}.
\end{proof}

\begin{proof}[Proof of \cref{cover-rev}]
For \cref{cover-rev-sub}, if $F(B)=0$ then there is nothing to prove, so we assume henceforth that $F(B)>0$. Therefore, also $F(A)\ge F(B)>0$ and thus $F|_A$ and $F|_A|_B=F|_B$ are well defined.
We begin by noting that $\Rev(F|_A) \ge F|_A(B) \cdot \Rev(F|_B)$. Indeed, this inequality holds since for any mechanism $\mech$ (in particular, any mechanism obtaining close to optimal revenue from~$F|_B$), we have $\Rev_{\mech}(F|_A)\ge F|_A(B)\cdot\Rev_{\mech}(F|_B)$; by definition of $\Rev$, the inequality follows. Therefore, we have that $F(A) \cdot \Rev(F|_A) \ge F(A) \cdot F|_A(B) \cdot \Rev(F|_B) = F(B) \cdot \Rev(F|_B)$, as required.

For \cref{cover-rev-cover}, we start by defining $B_i \eqdef A_i \setminus \bigcup_{j=1}^{i-1} A_j$. By definition, $(B_i)_{i=1}^m$ is a partition of~$\RRn$. We first claim that $\Rev(F) \le \sum_{i=1}^{m} F(B_i) \cdot \Rev(F|_{B_i})$. This is the ``subdomain stitching'' lemma of \cite{BILW}; for completeness, we will repeat the one-sentence proof:
for any mechanism $\mech$ (in particular, any mechanism obtaining close to optimal revenue from $F$), we have that $\Rev_{\mech}(F)=\sum_{i=1}^m F(B_i)\cdot\Rev_{\mech}(F|_{B_i})$;\footnote{Similarly, if $F(B_i)=0$ for some $i\in[n]$, then we define $F(B_i)\cdot\Rev_{\mech}(F|_{B_i})$ to equal~$0$.} by definition of $\Rev$, the inequality follows. Now, by \cref{cover-rev-sub}, we have that ${F(B_i)\cdot\Rev(F|_{B_i})} \le {F(A_i)\cdot\Rev(F|_{A_i})}$ for every $i\in[m]$. Combining both of these, we obtain that
$\Rev(F) \le \sum_{i=1}^{m} {F(B_i) \cdot \Rev(F|_{B_i})} \le \sum_{i=1}^{m} F(A_i) \cdot \Rev(F|_{A_i})$,
as required.
\end{proof}

\subsection{Exclusivity at Expensive Menu Entries}\label{expensive-exc-sec}

Having proven \cref{no-two-high}, we phrase and prove the next steps for arbitrary exclusively unbounded distributions, i.e., not necessarily product distributions conditioned upon~$\EU$.
As outlined above, our second step toward proving \cref{cne-upper-bound}, which we take in this \lcnamecref{expensive-exc-sec}, shows that in any mechanism over some exclusively unbounded distribution, the ``expensive'' part of the menu, i.e., the part of the menu consisting of all menu entries that cost more than some $\thresh=\poly(n,\nicefrac{1}{\varepsilon})$, can be simplified without significant loss in revenue to make the mechanism \emph{$\thresh$-exclusive}, i.e., to make each expensive menu entry only allocate a single item with non-zero probability. This subsection is dedicated to the statement and proof of \cref{expensive-exc}, which formalizes this step.

\begin{definition}[Exclusive Mechanism]
Let $n\in\mathbb{N}$ and let $\thresh\in\RR$. We say that an $n$-item mechanism is \emph{$\thresh$-exclusive} if it allocates (with positive probability) at most one item whenever it charges strictly more than $\thresh$.\footnote{While this also implies that the allocated item is quite expensive, and therefore some may say ``exclusive,'' the exclusivity discussed in the definition is that of solely this specific item being sold.}
\end{definition}

\begin{lemma}\label{expensive-exc}
Let $n\in\NN$ s.t.\ $n\ge2$, let $\bound\in\RR$, let $\varepsilon\in(0,1)$, and set $\thresh\eqdef\frac{4\cdot (n-1)\cdot\bound}{\varepsilon^2}$. For every $F\in\Delta(\EU)$ and for every IC and IR $n$-item mechanism $\mech$, there exists an $\thresh$-exclusive IC and IR $n$-item mechanism $\mech'$ such that $\Rev_{\mech'}(F) \ge (1-\varepsilon)\cdot\Rev_{\mech}(F)$.
\end{lemma}

\begin{proof}
Set $\epst\eqdef\nicefrac{\varepsilon}{2}$. We construct a new IC and IR mechanism $\mech'$ as follows:
\begin{itemize}
\item 
For every menu entry $e=(\vec{x};p)\in\mech$ with $p\le\thresh$, we add the menu entry $e$, unmodified, to~$\mech'$.
\item
For every menu entry $e=(\vec{x};p)\in\mech$ with $p>\thresh$, we add the following\footnote{Following standard notation, we use $(y',\vec{y}_{-i};q)$, for $\vec{y}\in[0,1]^n$, $i\in[n]$, $y'\in[0,1]$, and $q\in\RR$, to denote the outcome $(y_1,\ldots,y_{i-1},y',y_{i+1},\ldots,y_n;q)$, i.e., an outcome that is identical to $(\vec{y};q)$ in price and all winning probabilities, except the winning probability of item $i$, which is set to $y'$.} $n$ menu entries to~$\mech'$:
${\bigl(x_1,\vec{0}_{-1};(1-\epst)\cdot p\bigr)}$, ${\bigl(x_2,\vec{0}_{-2};(1-\epst)\cdot p\bigr)}$,$\cdots$,${\bigl(x_n,\vec{0}_{-n};(1-\epst)\cdot p\bigr)}$.
Each of these menu entries is a modified version of $e$ that completely ``unallocates'' all but one of the items, while giving a slight multiplicative price discount of $(1-\epst)$.
\end{itemize}
Finally, we define $\mech'$ to be the (topological) closure of the set of menu entries added above to $\mech'$.\footnote{Taking the closure ensures that a utility-maximizing entry with maximal price exists for every buyer type. See \cref{closure} for more details.}

By definition, $\mech'$ is $\thresh$-exclusive. We note that since $\mech$ is IR, it contains the menu entry $(\vec{0};0)$. Therefore, $\mech'$ also contains this menu entry, and hence $\mech'$ is IR as well. It remains to show that $\mech'$ obtains a revenue of at least $(1-\varepsilon)\cdot\Rev_{\mech}(F)$ from $F$.
Let us compare the payments that $\mech'$ and $\mech$ extract from a buyer of each type $v=(v_1,v_2,\ldots,v_n)\in\EU$. We reason by cases according to the menu entry of choice\footnote{If more than one utility-maximizing menu entry with maximal price exists, then here and whenever we henceforth refer to the ``menu entry of choice'' of some buyer type, we choose one such entry arbitrarily.} of buyer type~$v$ from $\mech$, which we denote by $e=(\vec{x};p)$, and show that in either case, the payment extracted from a buyer of type~$v$ decreases by at most a multiplicative factor of $(1-2\epst)$ in $\mech'$ compared to $\mech$.
\begin{itemize}
\item
If $p\le\thresh$, then by definition $e\in\mech'$. We claim that $v$ weakly prefers $e$ to all menu entries $f'=(\vec{y}';q')\in\mech'$ with\footnote{\label{expensive-exc-implies-closure}As we show weak preference and as $q'$ is defined via a strict inequality, by continuity of the utility function the correctness of the claim for all $f'$ before taking the closure of $\mech'$ implies its correctness for all $f'$ in the closure as well.} $q'<(1-\epst)\cdot p$. Indeed, for every such menu entry $f'$, we have that $q'<(1-\epst)\cdot\thresh$, and so by definition we have that $f'\in\mech$, and so by definition of $e$ we have that $v$ weakly prefers $e$ to $f'$. Therefore, the price of the menu entry chosen by $v$ from $\mech'$ is at least $(1-\epst)\cdot p$, and so the payment extracted from a buyer of type $v$ decreases by at most a multiplicative factor of $(1-\epst)$ in~$\mech'$ compared to $\mech$.
\item
Otherwise, i.e., if $p>\thresh$, then since $\thresh>n\cdot\bound$ (since $n\ge2$ and $\varepsilon<1$), by IR there must exist $i\in[n]$ s.t.\ $v_i>\bound$. Since $v\in \EU$, we have that $v_k\le\bound$ for every $k\in[n]\setminus \{i\}$.
Let $e'=(\vec{x}';p')\eqdef\bigl(x_i,\vec{0}_{-i};(1-\epst)\cdot p\bigr)$ be the menu entry in~$\mech'$ corresponding to $e$ that unallocates all items except item~$i$.
We claim that $v$ weakly prefers~$e'$ to all menu entries $f'=(\vec{y}';q')\in\mech'$ with\footnote{See \cref*{expensive-exc-implies-closure}.} $q'<p'-\frac{(1-\epst)\cdot(n-1)\cdot\bound}{\epst}$. Let~$f'$ be such a menu entry and denote the menu entry corresponding to $f'$ in~$\mech$ by $f=(\vec{y};q)$ (where either $f'=f$ or $f'=(y_j,\vec{0}_{-j},(1-\epst)\cdot q)$ for some $j\in[n]$).
Recall that $u_e(v)\ge  u_f(v)$ by definition of $e$.
Noting that $q \le \frac{q'}{1-\epst} < p - \frac{(n-1)\cdot\bound}{\epst}$, we indeed get
\begin{multline*}
u_{e'}(v) = x'_i\cdot v_i - p' =
x_i \cdot v_i - p + \epst\cdot p \ge
-(n-1)\cdot\bound + \sum_{k=1}^n x_k\cdot v_k - p + \epst\cdot p = \\
= -(n-1)\cdot\bound + u_e(v) + \epst\cdot p \ge
-(n-1)\cdot\bound + u_f(v) + \epst\cdot p = \\
= -(n-1)\cdot\bound + \sum_{k=1}^n y_k\cdot v_k - q + \epst\cdot p >
-(n-1)\cdot\bound + \sum_{k=1}^n y_k\cdot v_k - q + \epst\cdot \left(q + \tfrac{(n-1)\cdot\bound}{\epst}\right) = \\
= -(n-1)\cdot\bound + \sum_{k=1}^n y_k\cdot v_k - q \cdot (1-\epst) + (n-1)\cdot\bound \ge
\sum_{k=1}^n y_k\cdot v_k - q' \ge \\
\ge \sum_{k=1}^n y'_k\cdot v_k - q' =
u_{f'}(v).
\end{multline*}
Therefore, the payment that $\mech'$ extracts from a buyer of type~$v$ is at least
\begin{multline*}
p'-\tfrac{(1-\epst)\cdot (n-1)\cdot\bound}{\epst} = p' - (1-\epst)\cdot\epst\cdot\thresh > p' - (1-\epst)\cdot\epst\cdot p \ge \\
\ge p' - \epst\cdot p' = (1-\epst) \cdot p' = (1-\epst)^2 \cdot p > (1-2\epst) \cdot p,
\end{multline*}
and so the payment extracted from a buyer of type~$v$ decreases by at most a multiplicative factor of $(1-2\epst)$ in $\mech'$ compared to $\mech$.
\end{itemize}

To summarize, the revenue from each buyer type $v\in\EU$ decreases by at most a multiplicative factor of $(1-2\epst)=(1-\varepsilon)$ in $\mech'$ compared to $\mech$, and so the (overall) revenue that $\mech'$ obtains from $\EU$ is at least a $(1-\varepsilon)$ fraction of the revenue that $\mech$ obtains from $\EU$, as required.
\end{proof}

\subsection{Trimming the Expensive Part of the Menu}\label{expensive-trim-sec}

As outlined above, our third step toward proving \cref{cne-upper-bound}, which we take in this \lcnamecref{expensive-trim-sec}, shows that in any exclusive mechanism over some exclusively unbounded distribution, the expensive part of the menu can be simplified without significant loss in revenue, so that it contains at most $2n$ menu entries. This subsection is dedicated to the statement and proof of \cref{expensive-trim}, which formalizes this step, which is the most technically elaborate of all steps of our proof of \cref{cne-upper-bound}.

\begin{lemma}\label{expensive-trim}
Let $n\in\NN$ s.t.\ $n\ge2$, let $\bound\in\RR$, let $\varepsilon\in(0,1)$, and let $\thresh\ge\Max\bigl\{n\cdot\bound+2, \frac{(n-1)\cdot\bound+1}{\varepsilon}\bigr\}$. For every $F\in\Delta(\EU)$ and for every $\thresh$-exclusive IC and IR $n$-item mechanism $\mech$, there exists an $\thresh$-exclusive IC and IR $n$-item mechanism $\mech'$, such that both of the following hold.
\begin{itemize}
\item
$\Rev_{\mech'}(F) \ge (1-\varepsilon)\cdot\Rev_{\mech}(F)$.
\item
(The set of menu entries in) $\mech'$ coincides with a subset of the set of menu entries in $\mech$ that cost at most $\thresh$, with the addition of at most $2n$ menu entries.
\end{itemize}
\end{lemma}

\begin{proof}
Throughout the proof, we assume w.l.o.g.\ that each menu entry in $\mech$ (with the possible exception of $(\vec{0};0)$) is chosen by at least one buyer type $v\in\EU$.\footnote{While $\mech$ is not necessarily closed after the removal of all menu entries that are not chosen by any buyer type, it does possess a utility-maximizing entry with maximal price for every buyer type. See \cref{closure} for a discussion.}
For every $i\in[n]$, we define \[\mech_i \eqdef \bigl\{(\vec{x};p)\in\mech ~\big|~ p>\thresh \And x_i>0\bigr\},\]
and
\[W_i\eqdef\{v\in\EU\mid\mbox{$v$'s menu entry of choice from $\mech$ is in $\mech_i$}\}.\]

As $\mech$ is $\thresh$-exclusive, we have that $x_j=0$ for every $j\in[n]\setminus\{i\}$ and $(\vec{x};p)\in \mech_i$.
Thus, $W_i$ is the set of buyer types that choose to pay more than $\thresh$, and receive in return a positive probability for winning item $i$ and zero probability for winning any other item.
Our goal is to apply the single-dimensional analysis of \cite{Myerson81} in order to replace the plethora of menu entries in each $\mech_i$ with a small constant number of menu entries.
\textbf{Until noted otherwise, fix $\mathbfit{i\in[n]}$ s.t.\ $\mathbfit{\mech_i\ne\emptyset}$.}

We define $b_i\eqdef\inf\bigl\{p~\big|~(\vec{x};p)\in \mech_i\bigr\}$ and $s_i\eqdef\inf\bigl\{x_i~\bigl|~(\vec{x};p)\in \mech_i\bigr\}$, and set $z_i\eqdef(s_i,\vec{0}_{-i};b_i)$. One may intuitively think of $z_i$ (charging $b_i$ for an $s_i$ probability of winning item $i$) as the cheapest entry, which also allocates the least probability, in $\mech_i$, although formally (since $\mech$ is not closed, and also due to the way in which $\mech_i$ is defined) $z_i$ need not necessarily be in $\mech_i$.\footnote{Nonetheless, due to the assumption that every menu entry is chosen by some buyer type, we are able to show that $z_i$ is in the closure of $\mech_i$; see \crefpart{mechi}{sequence} below.  This property is heavily used throughout our proof.} Our strategy is to show that, in a precise sense, $\mech_i$ behaves on $W_i$ as follows: first allocate the buyer $z_i$ (i.e., provide a ``starting winning probability'' of $s_i$ for item $i$, and charge a ``base price'' of $b_i$), and then hold a ``continuation mechanism'' for possibly allocating some or all of the remaining $(1-s_i)$ probability of winning item $i$. In this ``continuation mechanism,'' we allow the buyer to swap $z_i$ for a different entry from $\mech_i$, paying the difference in costs and increasing the probability of getting item $i$ accordingly. We now make this statement precise. \textbf{Until noted otherwise, fix $\mathbfit{i\in[n]}$ s.t.\ $\mathbfit{F(W_i)>0}$ (and so, by definition, $\mathbfit{\mech_i\ne\emptyset}$) and $\mathbfit{s_i<1}$.} As we will see below, we require that $s_i<1$ in order for this ``continuation mechanism'' to be well defined (and in fact make sense; otherwise there is no ``remaining probability'' to sell) and that $F(W_i)>0$ for the valuation distribution of that mechanism to be well defined. As we will show, if $s_i=1$ then $|\mech_i|=\bigl|\{z_i\}\bigr|=1$ (see \crefpart{mechi}{single} below), and so there will be no need to reduce the number of menu entries in $\mech_i$ in this case (and if $F(W_i)=0$, then the revenue of the original mechanism $\mech$ from $W_i$ is zero, and so we will be able to simply delete $\mech_i$ from the original mechanism, without replacing it with anything).

For every $n$-dimensional buyer type $v\in W_i$, we define the single-dimensional valuation for the ``remaining probability'' of winning item $i$ by \[\alpha(v)\eqdef v_i\cdot (1-s_i).\]
We define the corresponding single-dimensional buyer type space of our ``continuation mechanism'' as \[W_i^1 \eqdef\bigl\{\alpha(v) ~\big|~ v\in W_i\}\subseteq\RR,\] and define a distribution  $F^1_i\in\Delta(W_i^1)\subseteq\Delta(\RR)$ over it by \[F^1_i(V) \eqdef F|_{W_i}\bigl(\alpha^{-1}(V)\bigr)\] for every measurable set $V\subseteq W_i^1$. (Recall that $F(W_i)>0$; therefore, $F|_{W_i}$ is well defined.) Very roughly speaking, $F^1_i$ is defined such that its density at every $v^1\in W_i^1$ can be informally thought of as the sum of the densities of $F|_{W_i}$ at all $v\in W_i$ s.t.\ $\alpha(v)=v^1$.

Having defined the single-dimensional buyer type space of the ``continuation mechanism,'' we now turn to defining the (menu of the) mechanism itself. For an outcome (not necessarily a menu entry in~$\mech$) of the form $e=(x_i,\vec{0}_{-i};p)$, we define \[\beta(e)=\beta(x_i,\vec{0}_{-i};p)\eqdef\bigl(\tfrac{x_i-s_i}{1-s_i};p-b_i\bigr),\] i.e., an entry selling a fraction of the remaining probability $(1-s_i)$ so that (in addition to the starting winning probability $s_i$) the overall winning probability is $x_i$, in exchange for ``upping'' the price (from the base price $b_i$) to $p$. (Recall that $s_i<1$; therefore, $\beta$ is well defined.) We define the ``continuation mechanism'' by\footnote{We emphasize that we do not take the closure (see \cref{closure} for a discussion) of $\mech_i^1$, as there is no ``need'' to do so. Indeed, as shown by \crefpart{alphabeta}{choice} below, $\mech_i^1$ already possesses a utility-maximizing entry with maximal price for each buyer type.} \[\mech_i^1 \eqdef \bigl\{\beta(e) ~\big|~ e\in \mech_i \cup \{z_i\} \bigr\}.\] We note that $\mech_i^1 \in [0,1]\times\RR$ by definition of $s_i$ and $b_i$, and furthermore that $(0;0)=\beta(z_i)\in\mech_i^1$; therefore, $\mech_i^1$ is an IC and IR $1$-item mechanism. We now make precise the above claim regarding $\mech_i$ acting as if the buyer were allocated $z_i$ and then bid in the single-dimension ``continuation mechanism''~$\mech_i^1$, thus relating the revenue of the single-dimensional ``continuation mechanism'' $\mech_i^1$ with that of the original mechanism $\mech$ from $W_i$.

\begin{sublemma}[Relation between $\mech_i^1$ and $\mech$]\label{alphabeta}
\leavevmode
\begin{parts}
\item\label{alphabeta-choice}
For every $v\in W_i$, we have that
$\beta(e)$ is the menu entry of choice of $\alpha(v)$ from $\mech_i^1$, where $e\in\mech_i$ is the menu entry of choice of $v$ from $\mech$.
\item\label{alphabeta-revenue}
$\Rev_{\mech_i^1}(F^1_i) = \Rev_{\mech}(F|_{W_i}) - b_i$.
\end{parts}
\end{sublemma}

The proof of \cref{alphabeta} is given after the proof of \cref{expensive-trim}. We note that the only delicate part of the proof (apart from the definitions, of course) is showing that even if $z_i$ is not an entry of $\mech_i$, then in the notation of the \lcnamecref{alphabeta} no buyer type chooses $z_i$ over $e$ (equivalently, $(0;0)$ over $\beta(e)$); see \crefpart{no-drop}{z} below.

Now that we have formally constructed a suitable single-dimensional mechanism, we apply the celebrated theorem of \cite{Myerson81}, to obtain a price $c_i\in\RR$ s.t.\ the IC and IR $1$-item mechanism $\mech^{*}_i\eqdef\bigl\{(0;0),(1;c_i)\bigr\}$ maximizes the revenue from $F^1_i$, i.e., $\Rev_{\mech^{*}_i}(F^1_i)=\Rev(F^1_i)$. Ideally, we would have now wanted to replace all entries from $\mech_i$ in $\mech$ by the two entries $z_i=\beta^{-1}\bigl((0;0)\bigr)$ and $o_i\eqdef(1,\vec{0}_{-i};c_i+b_i)=\beta^{-1}\bigl((1;c_i)\bigr)$, however it is not clear that none of the buyer types in $W_i$ whose utility has dropped as a result of this replacement would not actually now choose some ``cheap'' entry instead, which might cause a sharp decrease in revenue. Therefore, in order to make sure that no buyer type has an incentive to switch to a very cheap entry, instead of replacing $\mech_i$ with $\{z_i,o_i\}$ we replace it with $\{z'_i,o'_i\}$, where $z'_i$ and $o'_i$ are slightly discounted versions of $z_i$ and $o_i$, respectively. Formally, we define $z'_i\eqdef\bigl(s_i,\vec{0}_{-i};b_i-(n-1)\cdot\bound-1\bigr)$ and $o'_i\eqdef\bigl(1,\vec{0}_{-i};c_i+b_i-(n-1)\cdot\bound-1\bigr)$. The menu entries $z'_i$  and $o'_i$ sell the same winning probabilities as $z_i$ and $o_i$ (i.e., $s_i$ and $1$), respectively, but with a slight (when compared to $E$) additive price discount of $(n-1)\cdot\bound+1$.

\textbf{Having defined the new menu entries for every $\mathbfit{i}$ separately}, we define our ``trimmed'' mechanism:\footnote{As before, we emphasize that we do not take the closure (see \cref{closure} for a discussion) of $\mech'$, as there is no ``need'' to do so. Indeed, $\mech'$ already possess a utility-maximizing entry with maximal price for each buyer type except perhaps for a set of buyer types of measure zero. For buyer types $v\in\EU\setminus\bigcup_{i\in[n]:F(W_i)=0~\mathrm{or}~s_i<1} W_i$ such an entry exists since the entry of choice of $v$ from $\mech$ is in $\mech'$ (as mentioned above, when $s_i=1$ we have that $\mech_i=\{z_i\}$; see \crefpart{mechi}{single} below) and only finitely many additional menu entries were added. For buyer types $v\in \bigcup_{i\in[n]:F(W_i)>0\And s_i<1}W_i$ we show in \crefpart{alphabeta-rev}{choice} below that such an entry exists.}
\[
\mech'\eqdef
\Bigl(\mech\cap\bigl([0,1]^n\times[0,\thresh]\bigr)\Bigr)
\cup \smashoperator[r]{\bigcup_{\substack{i\in[n]:\\F(W_i)>0\And\\s_i=1}}}\{z_i\}
\cup \smashoperator[r]{\bigcup_{\substack{i\in[n]:\\F(W_i)>0\And\\s_i<1}}}\{z'_i, o'_i\}.
\]
We note that since $\mech$ is IR, it contains the menu entry $(\vec{0};0)$, and therefore so does $\mech'$ and hence $\mech'$ is IR as well. By definition, $\mech'$ contains at most $2$ \emph{expensive} (i.e., costing more than $\thresh$) menu entries for each $i\in[n]$, for a total of at most $2n$ expensive menu entries, as required. It remains to reason about the revenue of~$\mech'$. To do so, we now formulate what can be thought of as a ``converse'' of \cref{alphabeta}, thus relating the revenue of the new trimmed mechanism~$\mech'$ from~$W_i$ with the revenue of the optimal single-dimensional \citeauthor{Myerson81} mechanism $\mech^{*}_i$.

\begin{sublemma}[Relation between $\mech'$ and $\mech^{*}_i$]\label{alphabeta-rev}
Let $i\in[n]$ s.t.\ $F(W_i)>0$ and $s_i<1$.
\begin{parts}
\item\label{alphabeta-rev-choice}
For every $v\in W_i$, the menu entry of choice of $v$ from $\mech'$ is either $o'_i$ or $z'_i$:
\begin{itemize}
\item
$o'_i$ is the menu entry of choice of $v$ from $\mech'$ iff $(1;c_i)$ is the menu entry of choice of~$\alpha(v)$ from $\mech^{*}_i$.
\item
$z'_i$ is the menu entry of choice of $v$ from $\mech'$ iff $(0;0)$ is the menu entry of choice of~$\alpha(v)$ from $\mech^{*}_i$.
\end{itemize}
\item\label{alphabeta-rev-revenue}
$\Rev_{\mech'}(F|_{W_i}) = \Rev_{\mech^{*}_i}(F^1_i)+b_i-(n-1)\cdot\bound-1$.
\end{parts}
\end{sublemma}

The proof of \cref{alphabeta-rev} is given after the proof of \cref{expensive-trim}. As mentioned above, the delicate part of the proof is showing that no buyer type suddenly chooses a very cheap menu entry (or, hypothetically, an expensive menu entry from a different ``expensive part'' of~$\mech'$) due to the move from $\mech$ with $\mech'$. Showing this strongly depends on the assumption that  every menu entry in $\mech$ is chosen by some buyer type; see \cref{no-drop}(\labelcref{no-drop-cheap},\labelcref{no-drop-expensive}) below.

We now have all the pieces of the puzzle needed to show that our trimmed mechanism loses no significant revenue from $W_i$ for every $i\in[n]$ s.t.\ $s_i<1$. The following \lcnamecref{expensive-trim-revenue} shows not only that, but also that from no other part of the buyer type space, ``wealthy'' or ``poor'', does our trimmed mechanism lose significant revenue compared to the original untrimmed mechanism $\mech$.

\begin{sublemma}[Comparison between the Revenues of $\mech'$ and $\mech$ from Restrictions of $F$]\label{expensive-trim-revenue}\leavevmode
\begin{parts}
\item\label{expensive-trim-revenue-expensive-lt}
For every $i\in[n]$ s.t.\ $F(W_i)>0$ and $s_i<1$, we have that $\Rev_{\mech'}(F|_{W_i}) \ge (1-\varepsilon)\cdot\Rev_{\mech}(F|_{W_i})$.
\item\label{expensive-trim-revenue-expensive-eq}
For every $i\in[n]$ s.t.\ $F(W_i)>0$ and $s_i=1$, we have that $\Rev_{\mech'}(F|_{W_i}) = \Rev_{\mech}(F|_{W_i})$.
\item\label{expensive-trim-revenue-cheap}
Letting $P\eqdef\EU\setminus\bigcup_{i\in[n]} W_i$, if $F(P)>0$, then $\Rev_{\mech'}(F|_P) \ge (1-\varepsilon)\cdot\Rev_{\mech}(F|_P)$.
\end{parts}
\end{sublemma}

The proof of \cref{expensive-trim-revenue} is given after the proof of \cref{expensive-trim}.
We conclude the proof of \cref{expensive-trim}, as by \cref{expensive-trim-revenue} we have that
\begin{multline*}
\Rev_{\mech'}(F) = F(P) \cdot \Rev_{\mech'}(F|_P) + \sum_{i=1}^n F(W_i) \cdot \Rev_{\mech'}(F|_{W_i}) \ge \\
\ge (1-\varepsilon) \cdot \Bigl(F(P) \cdot \Rev_{\mech}(F|_P) + \sum_{i=1}^n F(W_i) \cdot \Rev_{\mech}(F|_{W_i})\Bigr) = (1-\varepsilon)\cdot\Rev_{\mech}(F),
\end{multline*}
as required.
\end{proof}

\subsubsection{Proof of \cref{alphabeta,alphabeta-rev,expensive-trim-revenue}}

In this subsection, we prove \cref{alphabeta,alphabeta-rev,expensive-trim-revenue}. Before doing so, we first phrase and prove three auxiliary \lcnamecrefs{mechi}, which in fact encompass most of the delicate details and technical complexity in the proof of \cref{alphabeta,alphabeta-rev,expensive-trim-revenue}.

\begin{sublemma}[Properties of $\mech_i$]\label{mechi}Let $i\in[n]$ s.t.\ $\mech_i\ne\emptyset$.
\begin{sloppypar}
\begin{parts}
\item\label{mechi-trichotomy}
For every $e=(x_i,\vec{0}_{-i};p)\in \mech_i$ and $f=(y_i,\vec{0}_{-i};q)\in \mech_i$, either ~i)~$e=f$, or ~ii)~$x_i<y_i$ and $p<q$, or ~iii)~$x_i>y_i$ and $p>q$.
\item\label{mechi-attained}
If there exists $e=(x_i,\vec{0}_{-i};p)\in \mech_i$ s.t.\ $x_i=s_i$ (i.e., if $s_i$ is attained as a minimum rather than merely an infimum), then $p=b_i$, i.e., $z_i=e$ and so $z_i\in \mech_i$.
\item\label{mechi-single}
If $s_i=1$, then $\mech_i=\{z_i\}$.
\item\label{mechi-sequence}
For every $i\in[n]$, there exists a sequence $\bigl(z_i^m=(s_i^m,\vec{0}_{-i};b_i^m)\bigr)_{m=1}^{\infty}\subseteq\mech_i$ s.t.\ ~i)~$(s_i^m)_{m=1}^{\infty}$ is weakly decreasing, ~ii)~$(b_i^m)_{m=1}^{\infty}$ is weakly decreasing, and ~iii)~$\lim_{m\rightarrow\infty} z_i^m = z_i=(s_i,\vec{0}_{-i};b_i)$.
\end{parts}
\end{sloppypar}
\end{sublemma}

\begin{proof}
We start with \cref{mechi-trichotomy}. If $p=q$, then we claim that $x_i=y_i$; indeed, assuming w.l.o.g.\ that $x_i\ge y_i$ we note that if we had $x_i>y_i$, then all buyer types $v$ with $u_f(v)\ge0$ (and hence $v_i>0$) would have strictly preferred $e$ to $f$, as the former sells (for the same price!)\ a strictly larger probability for winning item~$i$, contradicting the assumption that each menu entry in $\mech$ is chosen by at least one buyer type. If $p<q$, then 
we claim that $x_i<y_i$; indeed, otherwise all buyer types would strictly prefer $e$ to $f$ since $e$ would sell a weakly larger probability for winning item $i$ for a strictly lower price, contradicting the assumption that $f$ is chosen by at least one buyer type. Similarly, if $p>q$, then $x_i>y_i$.

We proceed to \cref{mechi-attained}. If there exists $e=(x_i,\vec{0}_{-i};p)\in \mech_i$ s.t.\ $x_i=s_i$, then for every $s=(y_i,\vec{0}_{-i};q)\in \mech_i$ we have by definition that $y_i \ge s_i = x_i$ and so by \cref{mechi-trichotomy} that $q \ge p$, and so by definition $b_i \ge p$. Since $e \in \mech_i$, we also have by definition that $b_i \le p$, and hence $b_i=p$. Therefore, $z_i=e$, as required.

We proceed to \cref{mechi-single}. Assume that $s_i=1$ and let $e=(x_i,\vec{0}_{-i};p)\in\mech_i$. By definition $s_i\le x_i$; however, we also have that $s_i=1\ge x_i$, and so $s_i=x_i$. By \cref{mechi-attained}, we therefore have that $e=z_i$. As $\mech_i\ne\emptyset$ by assumption, the proof of \cref{mechi-single} is complete.

We move on to \cref{mechi-sequence}. If $s_i$ is attained as a minimum, then by \cref{mechi-attained}, $z_i\in \mech_i$ and we may set $z_i^m\eqdef z_i$ for all $m\in\NN$, and so $(z_i^m)_{m=1}^{\infty}$ trivially meets all of the desired requirements. It remains to analyze the case in which $s_i$ is not attained as a minimum but rather only as an infimum; assume henceforth, therefore, that this is the case.

By definition of $s_i$, there exists a sequence $\bigl(z_i^m=(s_i^m,\vec{0}_{-i};b_i^m)\bigr)_{m=1}^{\infty}\subseteq\mech_i$ s.t.\ $(s_i^m)_{m=1}^{\infty}$ is weakly decreasing and s.t.\ 
$\lim_{m\rightarrow\infty} s_i^m = s_i$. By \cref{mechi-trichotomy}, $(b_i^m)_{m=1}^{\infty}$ is weakly decreasing as well. It is enough, therefore, to show that $\lim_{m\rightarrow\infty} b_i^m = b_i$. As $(b_i^m)_{m=1}^{\infty}$ is weakly decreasing and bounded from below by $b_i$ (by definition of $b_i$), we have that it indeed converges and, moreover, that $\lim_{m\rightarrow\infty} b_i^m \ge b_i$; assume for contradiction that $\lim_{m\rightarrow\infty} b_i^m > b_i$. Therefore, by definition of~$b_i$ there exists $e=(x_i,\vec{0}_{-i};p)\in \mech_i$ s.t.\ $\lim_{m\rightarrow\infty} b_i^m > p$ and so $b_i^m > p$ for all $m\in\NN$. By \cref{mechi-trichotomy}, we therefore have that $s_i^m>x_i$ for all $m\in\NN$, and so $s_i \ge x_i$. As by definition also $s_i \le x_i$, we obtain that $s_i=x_i$, contradicting the assumption that $s_i$ is not attained as a minimum.
\end{proof}

\begin{sublemma}\label{prefer-when-gt}
Let $i\in[n]$ and let $e=(x_i,\vec{0}_{-i};p)$ and $f=(y_i,\vec{0}_{-i};q)$ be two outcomes (not necessarily menu entries in $\mech$) assigning zero probability for winning any item other than $i$, s.t.\ $x_i > y_i$. For every $v\in\RRn$, we have that both of the following hold.
\begin{itemize}
\item
$u_e(v)\ge u_f(v)$ iff $v_i \ge \frac{p-q}{x_i-y_i}$.
\item
$u_f(v)\ge u_e(v)$ iff $\frac{p-q}{x_i-y_i} \ge v_i$.
\end{itemize}
\end{sublemma}

\begin{proof}
$u_e(v)\ge u_f(v)$ $\Leftrightarrow$ $v_i\cdot x_i - p \ge v_i \cdot y_i - q$ $\Leftrightarrow$ $v_i \cdot (x_i-y_i) \ge p-q$ $\Leftrightarrow$ $v_i \ge \frac{p-q}{x_i-y_i}$, where in the last equivalence we used the fact that $x_i - y_i > 0$. The proof of the second statement is identical, with all inequalities flipped.
\end{proof}

\begin{sublemma}[Preferences of $W_i$]\label{no-drop}
Let $i\in[n]$ s.t.\ $\mech_i\ne\emptyset$, let $v\in W_i$, and let $e\in \mech_i$ be the menu entry of choice of $v$ from $\mech$.
\begin{parts}
\item\label{no-drop-z}
$u_e(v) \ge u_{z_i}(v)$.
\item\label{no-drop-cheap}
$u_{z'_i}(v) > u_f(v)$, for every $f\in \mech\setminus \mech_i$.
\item\label{no-drop-expensive}
$\min\bigl\{u_{e}(v),u_{z'_i}(v)\bigr\} > u_f(v)$, for every outcome $f$ of the form $(y_j,\vec{0}_{-j};q)$ s.t.\ $j\in[n]\setminus\{i\}$ and $q\ge \thresh-(n-1)\cdot\bound-1$ (where $f$ need not necessarily be a menu entry in $\mech$).
\end{parts}
\end{sublemma}

\begin{proof}
Let $(z_i^m)_{m=1}^{\infty}$ be as in \crefpart{mechi}{sequence}. We start with \cref{no-drop-z}. By definition of $e$, we have that $u_e(v) \ge u_{z_i^m}(v)$ for every $m\in\NN$. By continuity of $u$, we therefore also have that $u_e(v) \ge u_{\lim_{m\rightarrow\infty}z_i^m}(v) = u_{z_i}(v)$, as required.

We turn to \cref{no-drop-cheap}. Denote $e=(x_i,\vec{0}_{-i};p)\in \mech_i$ ($e$ is still the menu entry of choice of $v$ from $\mech$). If $x_i=s_i$, then by \crefpart{mechi}{attained}, $e=z_i$, and the proof of \cref{no-drop-cheap} is complete, as $v$ (like any other buyer type) strictly prefers $z'_i$ to $z_i=e$, which is the menu entry of choice of~$v$ from $\mech$ and, therefore, is weakly preferred by $v$ to $f$. Otherwise, $x_i>s_i$, and so there exists $N\in\NN$ s.t.\ $x_i>s_i^m$ for every $m\ge N$; assume w.l.o.g.\ that $N=1$.

Denote $f=(\vec{y};q)$ and let $f'\eqdef(y_i,\vec{0}_{-i};q)$ (we emphasize that $f'$ need not necessarily be a menu entry in $\mech$). We claim that $u_{z_i^m}(v) \ge u_{f'}(v)$ for every $m\in\NN$. Indeed, let $m\in\NN$ and let
$w\in W_i$ be a buyer type that chooses $z_i^m$ from $\mech$. 
As $x_i > s_i^m$, by \cref{prefer-when-gt} (applied twice) we have that $v_i \ge \frac{p-b_i^m}{x_i-s_i^m}\ge w_i$. We note that $s_i^m>y_i$. Indeed, if $q>\thresh$, then as $f\in\mech\setminus\mech_i$, we have that $f\in\mech_j$ for some $j\in[n]\setminus\{i\}$, and so (since $\mech$ is $\thresh$-exclusive) $y_i=0<s_i^m$; otherwise (i.e., if $q\le\thresh<b_i^m$), if $y_i\ge s_i^m$ then we would have $u_f(w) = w_i\cdot y_i - q > w_i\cdot s_i^m - b_i^m = u_{z_i^m}(w)$, contradicting the definition of $w$.
By definition, $w$ weakly prefers $z_i^m$ to $f$, and by definition of~$f'$, we have that $w$ (as well as all other buyer types) weakly prefers $f$ to $f'$. Therefore, $w$ weakly prefers $z_i^m$ to $f'$, and so, since $v_i\ge w_i$ and by \cref{prefer-when-gt}, we obtain that $v_i \ge w_i \ge \frac{b_i^m-q}{s_i^m-y_i}$. Therefore, using \cref{prefer-when-gt} once more, we have that $v$ weakly prefers $z_i^m$ to $f'$, as claimed.
By continuity of $u$, we therefore have that $u_{z_i}(v) = u_{\lim_{m\rightarrow\infty}z_i^m}(v) \ge u_{f'}(v)$ as well.

As $\mech$ is IR and $e\in\mech_i$, we have that $v_i>\thresh>\bound$.
Hence, since $v\in\EU$ we have that $v_j\le\bound$ for every $j\in[n]\setminus\{i\}$. Therefore, $u_{z'_i}(v) = u_{z_i}(v) + (n-1)\cdot\bound+1 \ge u_{f'}(v) + (n-1)\cdot\bound+1 = u_f(v) - \sum_{j\in[n]\setminus\{i\}}v_j\cdot y_j + (n-1)\cdot\bound+1 > u_f(v)$, completing the proof of \cref{no-drop-cheap}.

We conclude by proving \cref{no-drop-expensive}. Recall that $\mech$ is IR; therefore, we have by definition of $v$ that $u_e(v) \ge 0$, and since $(\vec{0};0)\in \mech\setminus\mech_i$ we also have by \cref{no-drop-cheap} that $u_{z'_i}(v) \ge u_{(\vec{0};0)}(v) = 0$. Therefore, to prove \cref{no-drop-expensive} it suffices to show that $u_f(v)<0$. Indeed, recalling that $v_j\le\bound$, we have that $u_f(v) = v_j\cdot y_j - q \le \bound - q \le \bound - \thresh + (n-1)\cdot\bound+1 = n\cdot\bound+1 - \thresh < 0$.
\end{proof}

\begin{proof}[Proof of \cref{alphabeta}]
Recall that $i\in[n]$ s.t.\ $F(W_i)>0$ and $s_i<1$.
We start with \cref{alphabeta-choice}. We first claim that $u_{\beta(f)}\bigl(\alpha(v)\bigr)=u_f(v)-u_{z_i}(v)$ for every $f\in\mech_i\cup\{z_i\}$. Indeed,
denoting $f=(y_i,\vec{0}_{-i};q)$, we have
\begin{multline*}
u_{\beta(f)}\bigl(\alpha(v)\bigr) =
v_i\cdot(1-s_i)\cdot\tfrac{y_i-s_i}{1-s_i}-q+b_i =
v_i\cdot(y_i-s_i)-q+b_i = \\
= (v_i\cdot y_i-q) - (v_i\cdot s_i - b_i) =
u_f(v) - u_{z_i}(v).
\end{multline*}
For every $\beta(f),\beta(g)\in\mech_i^1$ (where $f,g\in\mech_i\cup\{z_i\}$), we therefore have that
\[u_{\beta(f)}\bigl(\alpha(v)\bigr) \ge u_{\beta(g)}\bigl(\alpha(v)\bigr) \Leftrightarrow
u_f(v) - u_{z_i}(v) \ge u_g(v) - u_{z_i}(v) \Leftrightarrow
u_f(v) \ge u_g(v).\]
Therefore, denoting the set of utility-maximizing entries for $v$ in $\mech_i\cup\{z_i\}$ by $M$, we have that the set of utility-maximizing entries for $v$ in $\mech_i^1$ is $M^1\eqdef\bigl\{\beta(f) \mid f\in M\bigr\}$. As $\beta$ is strictly monotone in the price coordinate, and since $e$ is an entry in $M$ with maximal price (indeed, $e$ is surpassed by $z_i$ neither in utility for $v$ (by \crefpart{no-drop}{z}) nor in price (as $p\ge b_i$ by definition)), we have that $\beta(e)$ is an entry in $M^1$ with maximal price, as required.

We proceed to \cref{alphabeta-revenue}. For every $v\in W_i$, denoting the menu entry of choice of $v$ from $\mech$ by $e=(x_i,\vec{0}_{-i};p)\in\mech_i$, we have by \cref{alphabeta-choice} that the payment that $\mech_i^1$ extracts from a buyer of type~$\alpha(v)$ is the price of $\beta(e)$, which equals $p-b_i$, i.e., precisely the payment that $\mech$ extracts from a buyer of type $v$, minus $b_i$. By definition of $F_i^1$, the proof is therefore complete.
\end{proof}

\begin{proof}[Proof of \cref{alphabeta-rev}]
We start with \cref{alphabeta-rev-choice}.
By \crefpart{no-drop}{cheap}, $v$ strictly prefers $z'_i$ to every menu entry in $\mech\setminus\mech_i$ (which, by \crefpart{mechi}{single}, also contains $\mech_j$ for every $j\in[n]$ s.t.\ $F(W_j)>0$ and $s_j=1$). By \crefpart{no-drop}{expensive}, $v$ strictly prefers $z'_i$ to $z'_j$ and $o'_j$ for every $j\in[n]\setminus\{i\}$ s.t.\ $F(W_j)>0$ and $s_j<1$. Therefore, by definition of $\mech'$, the menu entry of choice of $v$ from $\mech'$ is either $z'_i$ or $o'_i$. (By definition, the menu entry of choice of $\alpha(v)$ from $\mech^{*}_i$ is either~$(0;0$) or $(1;c_i)$.) We note that
\begin{multline*}
u_{o'_i}(v) = v_i-c_i-b_i+(n-1)\cdot\bound+1 = \\
= (v_i\cdot s_i - b_i + (n-1)\cdot\bound+1) + (v_i\cdot(1-s_i) - c_i) = u_{z'_i}(v) + u_{(1;c_i)}\bigl(\alpha(v)\bigr),
\end{multline*}
and so $v$ weakly (resp.\ strictly) prefers $o'_i$ to $z'_i$ iff $\alpha(v)$ weakly (resp.\ strictly) prefers $(1;c_i)$ to $(0;0)$, completing the proof of \cref{alphabeta-rev-choice} as the price of $o'_i$ is higher than that of $z'_i$ and the price of $(1;c_i)$ is higher than that of $(0;0)$ (and so in the case of indifference, $v$ chooses $o'_i$ and $\alpha(v)$ chooses $(1;c_i)$).

We proceed to \cref{alphabeta-rev-revenue}. For every $v\in W_i$, we have by \cref{alphabeta-rev-choice} that the payment that $\mech'$ extracts from a buyer of type $v$ equals precisely the payment that $\mech^{*}_i$ extracts from a buyer of type $\alpha(v)$, plus $b_i-(n-1)\cdot\bound-1$. By definition of $F_i^1$, the proof is therefore complete.
\end{proof}

\begin{proof}[Proof of \cref{expensive-trim-revenue}]
We start with \cref{expensive-trim-revenue-expensive-lt}. By \crefpart{alphabeta-rev}{revenue}, the definition of $\mech^{*}_i$ as the revenue-maximizing mechanism for $F^1_i$, and \crefpart{alphabeta}{revenue}, we have
\begin{multline*}
\Rev_{\mech'}(F|_{W_i}) = \Rev_{\mech^{*}_i}(F^1_i)+b_i-(n-1)\cdot\bound-1 \ge
\Rev_{\mech_i^{1}}(F^1_i)+b_i-(n-1)\cdot\bound-1 = \\
= \Rev_{\mech}(F|_{W_i})-(n-1)\cdot\bound-1 \ge
\Rev_{\mech}(F|_{W_i})-\varepsilon\cdot\thresh \ge
(1-\varepsilon)\cdot\Rev_{\mech}(F|_{W_i}),
\end{multline*}
where the last inequality is since $\Rev_{\mech}(F|_{W_i})\ge\thresh$ by definition of $W_i$.

We proceed to \cref{expensive-trim-revenue-expensive-eq}. Since $s_i=1$, by \crefpart{mechi}{single} we have that $\mech_i=\{z_i\}$, and so by definition, $z_i$ is the menu entry of choice of every $v\in W_i$ from $\mech$ (which, by \crefpart{mechi}{single}, also contains $\mech_j$ for every $j\in[n]\setminus\{i\}$ s.t.\ $F(W_j)>0$ and $s_j=1$). As by definition $z_i\in\mech'$, it is enough to show that $v$ strictly prefers $z_i$ to all menu entries in $\mech'\setminus\mech$. Indeed, by \crefpart{no-drop}{expensive}, $v$~strictly prefers $z_i$ to both $z'_j$ and $o'_j$ for every $j\in[n]$ s.t.\ $F(W_j)>0$ and $s_j<1$ (note that $j\ne i$ since $s_j<1=s_i$), i.e., $v$ strictly prefers $z_i$ to every menu entry in $\mech'\setminus\mech$, and so the proof of \cref{expensive-trim-revenue-expensive-eq} is complete.

We conclude by proving \cref{expensive-trim-revenue-cheap}. Let $v\in P$ and let $e=(\vec{x};p)\in\mech$ be the menu entry of choice of~$v$ from $\mech$. By definition of $P$ and since $\mech$ is $\thresh$-exclusive, we have that $p\le\thresh$. Since by definition we have that $e\in\mech\cap\bigl([0,1]^n\times[0,\thresh]\bigr)\subseteq\mech'$, and since the price of every menu entry in $\mech'\setminus\mech$ is at least $\min_{i\in[n]:F(W_i)>0\And s_i<1}b_i-(n-1)\cdot\bound-1 \ge \thresh-(n-1)\cdot\bound-1$, we have that the payment that $\mech'$  extracts from a buyer of type $v$ is either $p>(1-\varepsilon)\cdot p$, or at least $\thresh-(n-1)\cdot\bound-1 \ge (1-\varepsilon)\cdot\thresh \ge (1-\varepsilon)\cdot p$, and so the proof of \cref{expensive-trim-revenue-cheap} is complete.
\end{proof}

\subsection{Discretizing the Cheap Part of the Menu}\label{cheap-disc-sec}

As outlined above, our fourth and final step toward proving \cref{cne-upper-bound}, which we take in this \lcnamecref{cheap-disc-sec}, shows that in any mechanism over some exclusively unbounded distribution, the ``cheap'' part of the menu can be simplified without significant loss in revenue and without increasing the number of menu entries in the expensive part, so that (if both the parameter $\bound$ defining the exclusive unboundedness of the distribution and the parameter $\thresh$ defining the cheap part of the menu are polynomial in $n$ and $\nicefrac{1}{\varepsilon}$) it contains at most $(\nicefrac{n}{\varepsilon})^{O(n)}$ menu entries. This subsection is dedicated to the statement and proof of \cref{cheap-disc}, which formalizes this step.\footnote{See \cref{logn} for a stronger version of this \lcnamecref{cheap-disc-log}, strengthening the upper bound of $(\nicefrac{n}{\varepsilon})^{O(n)}$ to~$(\frac{\log n}{\varepsilon})^{O(n)}$.}

\begin{lemma}\label{cheap-disc}
Let $n\in\NN$, let $\bound\in\RR$, let $\thresh\in\RR$, and let $\varepsilon\in(0,1)$. For every $F\in\Delta(\EU)$ and for every IC and IR $n$-item mechanism $\mech$, there exists an IC and IR $n$-item mechanism $\mech'$ such that all of the following hold.
\begin{itemize}
\item
$\Rev_{\mech'}(F) \ge (1-\varepsilon)\cdot\Rev_{\mech}(F)-\varepsilon$.
\item
The menu entries that cost more than $(1-\varepsilon)\cdot\thresh$ in $\mech'$ are precisely the menu entries that cost more than $\thresh$ in $\mech$, each given a multiplicative price discount of $(1-\varepsilon)$. In particular, there are as many menu entries that cost more than $(1-\varepsilon)\cdot\thresh$ in $\mech'$ as there are that cost more than $\thresh$ in~$\mech$.
\item
There are fewer than $n\cdot \bigl\lceil\frac{n\cdot\bound}{\varepsilon^2}+1\bigr\rceil^{n-1}\cdot\bigl\lceil\frac{n\cdot(1-\varepsilon)\cdot\thresh}{\varepsilon^2}+1\bigr\rceil$ entries in $\mech'$ that cost at most $(1-\varepsilon)\cdot\thresh$.
\end{itemize}
\end{lemma}

\begin{proof}
We start by defining an ``interim'' mechanism $\mech''$ that will help us define the required mechanism~$\mech'$. While $\mech''$ may contain infinitely many menu entries, we show below that w.l.o.g.\ only finitely many of them are in fact chosen by any buyer type; consequently, we will derive the required mechanism~$\mech'$ by defining it to be the mechanism offering precisely these menu entries.

\begin{sloppypar}
Let $X\eqdef\bigl\lceil\frac{n\cdot\bound}{\varepsilon^2}\bigr\rceil$ and $P\eqdef\bigl\lceil\frac{n\cdot (1-\varepsilon)\cdot\thresh}{\varepsilon^2}\bigr\rceil$.
Let $\chi\eqdef\nicefrac{1}{X}\le\frac{\varepsilon^2}{n\cdot\bound}$ and $\psi\eqdef\frac{(1-\varepsilon)\cdot\thresh}{P}\le\frac{\varepsilon^2}{n}$.
For every $\delta>0$ and for every real number $r\in\RR$, we denote the rounding-down of $r$ to the ``$\delta$-grid'' by
\[
\lfloor r\rfloor_{\delta}\eqdef\delta\cdot\bigl\lfloor\nicefrac{r}{\delta}\bigr\rfloor.
\]
We furthermore denote the coordinate-wise rounding-down of every vector $\vec{r}=(r_1,\ldots,r_n)\in\RRn$ to the $\delta$-grid by $\lfloor\vec{r}\rfloor_{\delta}\eqdef\bigl(\lfloor r_1\rfloor_{\delta},\lfloor r_2\rfloor_{\delta},\ldots,\lfloor r_n\rfloor_{\delta}\bigr)$. We construct a new IC and IR mechanism $\mech''$ as follows:
\begin{itemize}
\item
For every menu entry $e=(\vec{x};p)\in\mech$ with $p\le\thresh$, we define $p'\eqdef (1-\varepsilon)\cdot p$ (the price of $e$, after a slight multiplicative discount of $(1-\varepsilon$)), and add the following $n$ menu entries to $\mech''$:
$\bigl(x_1,(\lfloor\vec{x}\rfloor_{\chi})_{-1};\lfloor p'\rfloor_{\psi}\bigr), \bigl(x_2,(\lfloor\vec{x}\rfloor_{\chi})_{-2};\lfloor p'\rfloor_{\psi}\bigr),\cdots,\bigl(x_n,(\lfloor\vec{x}\rfloor_{\chi})_{-n};\lfloor p'\rfloor_{\psi}\bigr)$. Each of these menu entries is a modified version of a discounted $e$ (i.e., $e$ with the price modified to $p'$) that rounds down all but one of the allocation probabilities to the $\chi$-grid, and rounds down the (discounted) price to the $\psi$-grid.
\item
For every menu entry $e=(\vec{x};p)\in\mech$ with $p>\thresh$, we add the menu entry $(\vec{x};p')$ to $\mech''$, where once again $p'\eqdef(1-\varepsilon)\cdot p$. This menu entry is a discounted version of~$e$ by the same slight multiplicative discount of $(1-\varepsilon)$ as above (but without any rounding).
\end{itemize}
\end{sloppypar}
By slight abuse of notation, we write $\lfloor r\rfloor_{0}\eqdef r$ for every $r\in\RR$, and $\lfloor\vec{r}\rfloor_{0}\eqdef r$ for every $\vec{r}\in\RRn$. Using this notation, we note that every menu entry that we have added to $\mech''$ is of the form $\bigl(x_i,(\lfloor\vec{x}\rfloor_{\chi'})_{-i};\lfloor(1-\varepsilon)\cdot p\rfloor_{\psi'}\bigr)$, for some $(\vec{x};p)\in\mech$, $i\in[n]$, and $(\chi',\psi')\in\bigl\{(\chi,\psi),(0,0)\bigr\}$.
Finally, we define $\mech''$ to be the closure of the set of menu entries added above to $\mech''$.\footnote{As before, taking the closure ensures that a utility-maximizing entry with maximal price exists for every buyer type. See \cref{closure} for a more details.}
We note that since $\mech$ is IR, it contains the menu entry $(\vec{0};0)$. Therefore (as $\lfloor 0\rfloor_{\delta}=0$ for every $\delta$), $\mech''$ also contains this menu entry, and hence $\mech''$ is IR as well.

We claim that $\Rev_{\mech''}(F) \ge (1-\varepsilon)\cdot\Rev_{\mech}(F)-\varepsilon$. Indeed, let us compare the payments that $\mech''$ and $\mech$ extract from a buyer of each type $v=(v_1,\ldots,v_n)\in\EU$.  Let $e=(\vec{x};p)\in\mech$ be the menu entry of choice of $v$ from $\mech$.
As $v\in\EU$, there exists $i\in[n]$ s.t.\ $v_k\le\bound$ for every $k\in[n]\setminus \{i\}$.
Let $e'\eqdef\bigl(x_i,(\lfloor\vec{x}\rfloor_{\chi'})_{-i};\lfloor(1-\varepsilon)\cdot p\rfloor_{\psi'}\bigr)\in\mech''$ be the menu entry in $\mech''$ corresponding to $e$ that does not round the allocation probability of item~$i$ (if $p\le\thresh$, then $(\chi',\psi')=(\chi,\psi)$; otherwise, $(\chi',\psi')=(0,0)$).
We claim that $v$ weakly prefers~$e'$ to all menu entries $f'=\bigl(y_j,(\lfloor\vec{y}\rfloor_{\chi''})_{-j};\lfloor(1-\varepsilon)\cdot q\rfloor_{\psi''}\bigr)\in\mech''$ with\footnote{As we show weak preference and as $q$ is defined via a strict inequality, by continuity of the utility function the correctness of the claim for all $f'$ before taking the closure of $\mech''$ implies its correctness for all $f'$ in the closure as well.} $q < p - \varepsilon$ (where $j\in[n]$, $f=(\vec{y};q)\in \mech$, and $(\chi'',\psi'')\in\bigl\{(\chi,\psi),(0,0)\bigr\}$). Indeed, we have that
\begin{multline*}
u_{e'}(v) = v_i \cdot x_i + \smashoperator[r]{\sum_{k\in[n]\setminus\{i\}}} v_k\cdot \lfloor x_k\rfloor_{\chi'} - \lfloor(1-\varepsilon)\cdot p\rfloor_{\psi'} \ge
\sum_{k=1}^n v_k\cdot x_k - (n-1)\cdot\bound\cdot\chi' - (1-\varepsilon)\cdot p = \\
= u_e(v) - (n-1)\cdot\bound\cdot\chi' + \varepsilon\cdot p \ge
u_f(v) - (n-1)\cdot\bound\cdot\chi' + \varepsilon\cdot p =
\sum_{k=1}^n v_k\cdot y_k - q - (n-1)\cdot\bound\cdot\chi' + \varepsilon\cdot p \ge \\
\ge v_j\cdot y_j + \smashoperator[r]{\sum_{k\in[n]\setminus\{j\}}} v_k\cdot\lfloor y_k\rfloor_{\chi''} - \lfloor(1-\varepsilon)\cdot q\rfloor_{\psi''} - \varepsilon\cdot q - \psi'' - (n-1)\cdot\bound\cdot\chi' + \varepsilon\cdot p = \\
= u_{f'}(v) - \varepsilon\cdot q - \psi'' - (n-1)\cdot\bound\cdot\chi' + \varepsilon\cdot p =
u_{f'}(v) + \varepsilon\cdot (p-q) - \psi'' - (n-1)\cdot\bound\cdot\chi' \ge \\
\ge u_{f'}(v) + \varepsilon\cdot (p-q) - \tfrac{\varepsilon^2}{n} - (n-1)\cdot\tfrac{\varepsilon^2}{n} =
u_{f'}(v) + \varepsilon\cdot (p-q) - \varepsilon^2 >
u_{f'}(v) + \varepsilon^2 - \varepsilon^2 =
u_{f'}(v).
\end{multline*}
Therefore, the price of the menu entry chosen by $v$ from $\mech''$ is at least $(1-\varepsilon)(p-\varepsilon)-\psi \ge (1-\varepsilon)\cdot(p-\varepsilon) - \varepsilon^2=
(1-\varepsilon)\cdot p - \varepsilon$, and so the payment extracted from a buyer of type $v$ in~$\mech''$ compared to $\mech$ decreases by at most a multiplicative factor of $(1-\varepsilon)$ followed by an additive decrease of at most $\varepsilon$. Overall, we therefore obtain that $\Rev_{\mech''}(F) \ge (1-\varepsilon)\cdot\Rev_{\mech}(F)-\varepsilon$, as claimed.

As mentioned above, while $\mech''$ may contain infinitely many menu entries, we now show that w.l.o.g.\ only finitely many of these menu entries are in fact chosen by any buyer type, and define the mechanism~$\mech'$ to be the mechanism offering precisely these menu entries.
Let $L\eqdef\bigcup_{i\in[n]}\bigl(\{i\}\times\{0,\chi,{2\cdot\chi},\ldots,X\cdot\chi\}^{[n]\setminus\{i\}}\bigr)\times\bigl\{0,\psi,2\cdot\psi,\ldots,P\cdot\psi\bigr\}\subset\bigcup_{i=1}^n\bigl(\{i\}\times[0,1]^{[n]\setminus\{i\}}\bigr)\times[0,\thresh]$.
By definition, every menu entry in $\mech''$ that costs at most $(1-\varepsilon)\cdot\thresh$ is of the form $(x_i,\vec{x}_{-i};p)$ for some $(i,\vec{x}_{-i},p)\in L$ and $x_i\in[0,1]$. For every $(i,\vec{x}_{-i},p)\in L$, we set $S_{(i,\vec{x}_{-i},p)} \eqdef \bigl\{x_i ~\big|~ (x_i,\vec{x}_{-i};p)\in\mech''\bigr\}\subseteq[0,1]$, and if $S_{(i,\vec{x}_{-i},p)}\ne\emptyset$, we also set\footnote{Since $\mech''$ is closed, $S_{(i,\vec{x}_{-i},p)}$ is compact.} $s_{(i,\vec{x}_{-i},p)}\eqdef\Max S_{(i,\vec{x}_{-i},p)}$. We define a new IC and IR mechanism $\mech'$ as follows:
\begin{multline*}
\mech'\eqdef \Bigl\{\bigl(s_{(i,\vec{x}_{-i},p)},\vec{x}_{-i};p\bigr) ~\Big|~ (i,\vec{x}_{-i},p)\in L \And S_{(i,\vec{x}_{-i},p)} \ne \emptyset \Bigr\} ~\cup~ \\*
\Bigl\{\bigl(\vec{x};(1-\varepsilon)\cdot p\bigr) ~\Big|~ (\vec{x};p)\in\mech \And p>\thresh\Bigr\}.
\end{multline*}
By definition, $\mech'$ is a subset of $\mech''$, obtained by removing from $\mech''$ only menu entries that w.l.o.g.\ no buyer type chooses. (Indeed, a buyer type that chooses from~$\mech''$ some entry $(x;p)\in\mech''$ weakly prefers $(s_{(i,\vec{x}_{-i},p)},\vec{x}_{-i};p)$ to $(x_i,\vec{x}_{-i};p)=(x;p)$, since by definition $x_i\in S_{(i,\vec{x}_{-i},p)}$. As both of these entries have the same price, we can assume w.l.o.g.\ that this buyer type in fact chooses $\bigl(s_{(i,\vec{x}_{-i},p)},\vec{x}_{-i};p\bigr)$ from $\mech''$.) Therefore, every buyer type chooses from $\mech'$ the same menu entry (or at least a menu entry with the same price) as from~$\mech''$, and so $\Rev_{\mech'}(F) = \Rev_{\mech''}(F) \ge (1-\varepsilon)\cdot\Rev_{\mech}(F)-\varepsilon$.

We conclude the proof by noting that, as required, there are fewer than $|L| = n\cdot(X+1)^{n-1}\cdot (P+1) = n\cdot \bigl\lceil\frac{n\cdot\bound}{\varepsilon^2}+1\bigr\rceil^{n-1}\cdot\bigl\lceil\frac{n\cdot(1-\varepsilon)\cdot\thresh}{\varepsilon^2}+1\bigr\rceil$ menu entries (not including $(\vec{0};0)$) that cost at most $(1-\varepsilon)\cdot\thresh$ in~$\mech'$, and that the menu entries that cost more than $(1-\varepsilon)\cdot\thresh$ in $\mech'$ are precisely the menu entries that cost more than $\thresh$ in $\mech$, each given a multiplicative price discount of $(1-\varepsilon)$.
\end{proof}

\subsection{Connecting the Dots}\label{connect-dots-sec}

We are now ready to ``connect the dots'' and use \cref{no-two-high,expensive-exc,expensive-trim,cheap-disc} to prove the qualitative version of \cref{cne-upper-bound}, which states that $C(n,\varepsilon)$ is finite for every number of items~$n\in\NN$ and $\varepsilon>0$, and moreover, to prove a weaker version of the quantitative version of that theorem, i.e., that~$C(n,\varepsilon)\le(\nicefrac{n}{\varepsilon})^{O(n)}$.

\begin{proof}
Let $n\in\NN$, let $\varepsilon\in(0,1)$, and let $F=F_1\times\cdots\times F_n\in\Delta(\RR)^n$. If $n=1$, then by the theorem of \cite{Myerson81}, we are done.
If $\Max_{i\in[n]} \Rev(F_i)=\infty$ (i.e., if there exists~$i\in[n]$ s.t.\ $\Rev(F_i)=\infty$), then by the same theorem of \cite{Myerson81}, we are done in this case as well, as a single menu entry suffices to get arbitrarily high revenue;\footnote{By convention, $(1-\varepsilon)\cdot\infty=\infty$.} otherwise, by \cref{sum-rev}, $\Rev(F)<\infty$.
If $\Max_{i\in[n]} \Rev(F_i)=0$, then the valuation of each item is $0$ with probability $1$ and so ${\Rev(F)=0}$ and we are done as well. Assume, therefore, that $\Max_{i\in[n]} \Rev(F_i)\in(0,\infty)$ and
set $\epst\eqdef\nicefrac{\varepsilon}{6}$.
By scaling the currency we assume w.l.o.g.\ that $R\eqdef\Max_{i\in[n]} \Rev(F_i)={(1-\epst)^{-5}}< (\nicefrac{6}{5})^5 < \nicefrac{5}{2}$. (This indeed is w.l.o.g.\ as our goal is to prove a multiplicative approximation.)

Let $\bound\eqdef \frac{2\cdot n\cdot(n-1)\cdot R}{\epst}$, and let $\thresh\eqdef\frac{4\cdot(n-1)\cdot\bound}{\epst^2}$. By \crefpart{no-two-high}{well-defined}, $F(\EU)>0$ and so $F|_{\EU}$ is well defined. Furthermore (e.g., by \crefpart{cover-rev}{sub}), $F(\EU)\cdot\Rev(F|_{\EU})\le\Rev(F)<\infty$, and so, as $F(\EU)>0$, we have that $\Rev(F|_{\EU})<\infty$. Therefore, by definition of $\Rev$, there exists an IC and IR $n$-item mechanism $\mech$ s.t.\ $\Rev_{\mech}(F|_{\EU})\ge(1-\epst)\cdot\Rev(F|_{\EU})$. By \cref{expensive-exc}, there exists an $\thresh$-exclusive IC and IR $n$-item mechanism $\mech'$ s.t.\ $\Rev_{\mech'}(F|_{\EU})\ge(1-\epst)\cdot\Rev_{\mech}(F|_{\EU})$. By \cref{expensive-trim}, there exists an IC and IR $n$-item mechanism $\mech''$ with at most $2n$ menu entries that cost more than $\thresh$, s.t.\ $\Rev_{\mech''}(F|_{\EU})\ge(1-\epst)\cdot\Rev_{\mech'}(F|_{\EU})$. By \cref{cheap-disc}, there exists an IC and IR $n$-item mechanism $\mech'''$ with at most $n\cdot \bigl\lceil\frac{n\cdot\bound}{\epst^2}+1\bigr\rceil^{n-1}\cdot\bigl\lceil\frac{n\cdot(1-\epst)\cdot\thresh}{\epst^2}+1\bigr\rceil$ menu entries that cost at most $(1-\epst)\cdot\thresh$ and (by definition of $\mech''$) at most $2n$ menu entries that cost more than $(1-\epst)\cdot\thresh$, s.t.\ $\Rev_{\mech'''}(F|_{\EU})\ge(1-\epst)\cdot\Rev_{\mech''}(F|_{\EU})-\epst$.
By definition of $\mech'''$, $\mech''$, $\mech'$, and~$\mech$, we have that
\begin{multline*}
\Rev_{\mech'''}(F|_{\EU})\ge
(1-\epst)\cdot\Rev_{\mech''}(F|_{\EU})-\epst\ge
(1-\epst)^2\cdot\Rev_{\mech'}(F|_{\EU})-\epst\ge \\
\ge (1-\epst)^3\cdot\Rev_{\mech}(F|_{\EU})-\epst\ge
(1-\epst)^4\cdot\Rev(F|_{\EU})-\epst\ge
(1-\epst)^5\cdot\Rev(F|_{\EU}),
\end{multline*}
where the last inequality is since $\Rev(F|_{\EU})\ge(1-\epst)\cdot\Rev(F) \ge (1-\epst)\cdot R = (1-\epst)^{-4}$ by \crefpart{no-two-high}{rev}.
Therefore, by \crefpart{no-two-high}{revm}, we have that \[
\Rev_{\mech'''}(F)\ge(1-\epst)^6\cdot\Rev(F) > (1-6\cdot\epst)\cdot\Rev(F)=(1-\varepsilon)\cdot\Rev(F).
\]

We conclude the proof by noting that the number of menu entries (not including $(\vec{0};0)$) in $\mech'''$ is less than 
\begin{multline*}
n\cdot \left\lceil\frac{n\cdot\bound}{\epst^2}+1\right\rceil^{n-1}\cdot\left\lceil\frac{n\cdot(1-\epst)\cdot\thresh}{\epst^2}+1\right\rceil+2n= \\
=n\cdot \left\lceil\frac{2\cdot n^2\cdot(n-1)\cdot R}{\epst^3}+1\right\rceil^{n-1}\cdot\left\lceil\frac{8\cdot n^2\cdot(n-1)^2\cdot(1-\epst)\cdot R}{\epst^5}+1\right\rceil+2n\le \\
\le n\cdot \left\lceil\frac{5\cdot n^2\cdot(n-1)}{\epst^3}+1\right\rceil^{n-1}\cdot\left\lceil\frac{20\cdot n^2\cdot(n-1)^2\cdot(1-\epst)}{\epst^5}+1\right\rceil+2n \le
(\nicefrac{n}{\varepsilon})^{O(n)}.\tag*{\qedhere}
\end{multline*}
\end{proof}

\subsection{A Stronger Upper Bound}\label{tighter}

The proof of \cref{cne-upper-bound} (in its full strength) strengthens the upper bound on~$C(n,\varepsilon)$ from $(\nicefrac{n}{\varepsilon})^{O(n)}$ to $\bigl(\frac{\log n}{\varepsilon}\bigr)^{O(n)}$, but at a cost of a more conceptually elaborate presentation. As this stronger bound carries the same qualitative message and is still exponential in $n$, we have decided, as noted above, to first present the proof of the weaker upper bound of $(\nicefrac{n}{\varepsilon})^{O(n)}$. We note that the only change required to prove this stronger bound is to \cref{cheap-disc}, where by following along the same lines as the proof above, but utilizing discretization techniques of \cite{hart-nisan-b} \cite[see also][]{DLN14} to discretize to a more carefully chosen logarithmic-sized grid --- which is achievable by not only rounding but also ``compensating'' between coordinates --- one may show an improved upper bound of $\bigl(\frac{\log n}{\varepsilon}\bigr)^{O(n)}$ on the number of menu entries into which the cheap part of the menu can be discretized in \cref{cheap-disc}. (Note that other than for the cheap part of the menu, the remaining \lcnamecrefs{expensive-trim} show that we only need at most~$2n$ more entries for the expensive part of the menu, and therefore no modification is required to these \lcnamecrefs{expensive-trim}, which either way contain most of the conceptual and technical ``beef'' of the proof, in order to show the stronger upper bound.) The statement of this stronger version of \cref{cheap-disc} is given as \cref{cheap-disc-log} in \cref{logn}, along with its full proof.

\subsection{Arbitrary Exclusively Unbounded Type Distributions}\label{exclusively-unbounded-sec}

We emphasize once more that in contrast to the first step (\cref{no-two-high}) of the proof of \cref{cne-upper-bound}, the remaining three steps (i.e., \cref{expensive-exc,expensive-trim,cheap-disc}, as well as the stronger version of the latter, \cref{cheap-disc-log}) hold for any exclusively unbounded type distribution, and not merely for one obtained by conditioning product distributions upon~$\EU$. Therefore, these steps allow us to derive results also for such distributions. Indeed, a proof similar to that of \cref{cne-upper-bound} (only without using \cref{no-two-high}) yields:

\begin{proposition}\label{upper-bound-correlated}
For every $n\in\NN$, $\bound\in\RR$, and $\varepsilon>0$, there exists  $C=C(n,\bound,\varepsilon)\le\bigl(\frac{\log n+\log \bound}{\varepsilon}\bigr)^{O(n)}$ such that for every (possibly even highly correlated) $F\in\Delta(\EU)$, we have that
$\Rev_C(F) \ge {(1-\varepsilon) \cdot \Rev(F) - \varepsilon}$.
\end{proposition}

\cref{upper-bound-correlated} states that a menu size of $\bigl(\frac{\log n+\log \bound}{\varepsilon}\bigr)^{O(n)}$ guarantees\footnote{Once again, following the proof above would give a weaker upper bound of $\bigl(\frac{n\cdot\bound}{\varepsilon}\bigr)^{O(n)}$, while replacing \cref{cheap-disc} with \cref{cheap-disc-log} from \cref{logn} gives the stronger upper bound.} a multiplicative-and-then-additive (with at most an $\varepsilon$ loss in each) approximation to the optimal revenue from any exclusively unbounded type distribution. If the $L^1$-norm of the support of this distribution is furthermore bounded away from zero, then so is the optimal revenue (since the optimal revenue is at least the lower bound on that norm, as the seller could price the grand bundle of all items at that lower bound and sell with probability~$1$), and so in this case the additive approximation from \cref{upper-bound-correlated} is absorbed into the multiplicative approximation, guaranteeing a purely multiplicative approximation of at least $(1-\varepsilon)$ for the revenue. (For example, this applies to any distribution over $\RR \times [1,\bound]^{n-1}$.) \cref{upper-bound-correlated} therefore generalizes a similar result by \cite{DLN14} \citep[which followed][who showed it for two items and gave a weaker result for more items]{hart-nisan-b} that shows the same upper bound of $\bigl(\frac{\log n + \log H}{\varepsilon}\bigr)^{O(n)}$ for guaranteeing a multiplicative approximation of at least $(1-\varepsilon)$, but only for distributions where the valuations of \emph{all} items are bounded, i.e., distributions over $[1,\bound]^n$.

Finally, we note that an exclusively unbounded valuation space is essentially the maximal valuation space for which a result along the lines of \cref{upper-bound-correlated} can be shown, since \cite{hart-nisan-b} show that for two items with unbounded valuations (and as a result, for any number of items where the valuations of at least two items are unbounded), no finite menu size can guarantee any fixed fraction of the optimal revenue. \cref{upper-bound-correlated} therefore shows that the use of distributions where the valuations of the two items grow arbitrarily large together \citep[like the distribution used in the proof of][]{hart-nisan-b} cannot be avoided in any proof of this impossibility result.


\section{A Small Menu for Item Pricing}\label{item-pricing}

In this \lcnamecref{item-pricing}, we prove \cref{poly-approx-srev}, which states that for every $\varepsilon>0$, there exists $d=d(\varepsilon)$ such that for every number of items $n\in\NN$
and $F_1,F_2,\ldots,F_n\in\Delta(\RR)$, we have for $C=n^d$ that
$\Rev_C(F_1 \times \cdots \times F_n) \ge (1-\varepsilon) \cdot \SRev(F_1 \times \cdots \times F_n)$. Concretely, we show that the \lcnamecref{poly-approx-srev} holds for $d\le O(\varepsilon^{-5})$.

\begin{definition}[$\SRev$]
Let $n\in\mathbb{N}$ be a number of items and let $F=F_1\times F_2\times\cdots\times F_n\in\Delta(\RR)^n$ be a product distribution over~$\RRn$. We denote the (expected) revenue obtainable from $F$ by selling each of the $n$ items separately via a revenue-maximizing mechanism by
\[\SRev(F)=\SRev(F_1\times F_2\times\cdots\times F_n)\eqdef\sum_{i=1}^n\Rev(F_i).\]
\end{definition}

\begin{remark}\label{srev-menu-size}
By the theorem of \cite{Myerson81}, the menu size (not including $(\vec{0};0)$) of the mechanism obtaining revenue $\SRev(F)$ by selling each of the $n$ items separately via a revenue-maximizing mechanism is at most $2^n-1$.
\end{remark}

The main idea underlying the proof of \cref{poly-approx-srev}, which we give in full detail below, is that instead of using $n$ separate mechanism (one for each item), which may result in an exponential-size menu, we use exponentially fewer separate mechanisms, which result in a polynomial-size menu. Recall that for every item $i$, by the theorem of \cite{Myerson81}, the optimal separate-selling revenue $\Rev(F_i)$ can be obtained via a take-it-or-leave-it offer for selling item $i$ (with probability~$1$) for a certain price $c_i$. Let $p_i$ be the probability item $i$ is sold (if it is offered for the price $c_i$). The separate mechanisms that we use are as follows:
\begin{itemize}
\item
For each item for which $c_i$ is very small (say, at most an $\nicefrac{\varepsilon}{n}$ fraction) compared to $\SRev(F)$, the optimal (separate-selling) revenue $\Rev(F_i)$ is very small compared to $\SRev(F)$ as well, and therefore the sum of the optimal revenues from all of these items is small compared to $\SRev(F)$. Therefore, we allow ourselves to not sell any of these items at all (or alternatively, give them to the buyer for free).
\item
For each item for which $c_i$ is very large (say, by at least a factor of $\nicefrac{n}{\varepsilon}$) compared to $\Rev(F_i)$, the probability that the item is sold (when selling items separately) is small. Hence, the probability that two or more of these items are sold is small. Therefore, we can afford to allow the buyer to buy at most only one of these items (for the same price $c_i$ as when selling each item separately) without incurring a significant loss in revenue. We therefore offer the buyer at most $n+1$ choices for these items.
\item
We partition the remaining items (those with ``nonextreme'' values of $c_i$) into $O(\log n)$ many bundles, each to be offered (via a separate mechanism) to the buyer at a take-it-or-leave-it price described below, where the ratio between the (optimal separate-selling) prices of any two items in a single bundle is small. Moreover, we show that this can be done s.t.\ each such bundle has $\sum p_i$ either very large or very small.
\begin{itemize}
\item
For each bundle with $\sum p_i$ very large, we show that the buyer's valuation of the bundle is tightly concentrated, allowing us to extract almost all of this valuation by offering the bundle for a price slightly below the expectation of this valuation.
\item
For each bundle with $\sum p_i$ very small, not unlike the case of the high-costing items above, we show that the probability of two or more of the items in the bundle being sold (when selling separately) is small. Since the prices of all of these items are similar, instead of allowing the buyer to buy at most one of these items (for its separate-selling price $c_i$) as in the case of the high-costing items above (thereby offering as many choices to the buyer for these items as there are items in the bundle), we simply offer the entire bundle for the cheapest (optimal separate-selling) price of any of the items in the bundle. We can afford to do so without incurring a significant loss in revenue since all of the items in the bundle have similar prices.
\end{itemize}
As each of these $O(\log n)$ many bundles is offered via a take-it-or-leave-it price, we offer the buyer altogether at most $\poly(n)$ choices for (buying any subset of) these bundles.
\end{itemize}

\begin{proof}[Proof of \cref{poly-approx-srev}]
Let $\varepsilon\in(0,1)$. We prove the theorem for\footnote{The following constant multiplier is used for ease of presentation, and is far from tight.} $d\eqdef\nicefrac{8194}{\varepsilon^5}$, and note that the mechanism that we construct is deterministic.

Let $n\in\NN$, and let $F=F_1\times F_2\times\cdots\times F_n\in\Delta(\RR)^n$.
For every $i\in[n]$, denote by $r_i\eqdef\Rev(F_i)$ the maximal revenue obtainable from $F_i$. By definition, $\SRev(F)=\SRev({F_1 \times \cdots \times F_n})=\sum_{i=1}^n r_i$.
If $n=1$, then by \cref{srev-menu-size} we are done, as selling the single item ``separately'' (and obtaining revenue $\SRev(F)$) requires at most $1=n^d$ menu entry (not including $(\vec{0};0)$).
If $2\le n<\nicefrac{4}{\varepsilon}$, then we are similarly done, as selling each item separately (and obtaining revenue $\SRev(F)$) requires at most $2^n-1<n^{\nicefrac{4}{\varepsilon}} < n^d$ menu entries. Assume henceforth, therefore, that $n\ge\nicefrac{4}{\varepsilon}$.
If $\Max_{i\in[n]}r_i=\infty$, i.e., if there exists $i\in[n]$ s.t.\ $r_i=\infty$, then we are done as well, as by the theorem of \cite{Myerson81} a single menu entry suffices to get arbitrarily high revenue.\footnote{As already noted above, by convention, $(1-\varepsilon)\cdot\infty=\infty$.}
If $\Max_{i\in[n]}r_i=0$, then there is nothing to prove.
Assume, therefore, that $\Max_{i\in[n]}r_i\in(0,\infty)$. By scaling the currency we assume w.l.o.g.\ that $\Max_{i\in[n]}r_i =1$. (This indeed is w.l.o.g.\ as our goal is to prove a multiplicative approximation.) Therefore, $\SRev(F) \ge 1$ (and also $\SRev(F)\le n$).

For every $i\in[n]$, by the theorem of \cite{Myerson81}, the revenue $r_i$ can be obtained from~$F_i$ via a take-it-or-leave-it offer for selling item $i$ (with probability $1$) for a certain price $c_i$. Let $p_i\eqdef\probge{v_i\sim F_i}{v_i}{c_i}$ be the probability that $c_i$ is accepted. We note that $r_i = p_i \cdot c_i$. As outlined above, instead of using $n$ separate mechanisms (one for each item), which may result in an exponential-size menu, we use exponentially fewer separate mechanisms, which results in a polynomial-size menu. To describe these separate mechanisms, we first partition the $n$ items into \emph{buckets} based on the optimal price $c_i$ for each item $i\in[n]$, and then describe the separate mechanism that we use for the items in each bucket. Set~$\epst\eqdef\nicefrac{\varepsilon}{4}$.
\begin{itemize}
\item
The \emph{low} bucket $L$ includes all items $i\in[n]$ for which $c_i < \nicefrac{\epst}{n}$.
\item
The \emph{high} bucket $H$ includes all items $i\in[n]$ for which $c_i \ge \nicefrac{n}{\epst}$.
\item
Let $m\eqdef\bigl\lceil\log_{1+\epst}\nicefrac{n}{\epst}\bigr\rceil$. We partition the remaining items (i.e., the items that are not already in the low or the high bucket) into $2m$ \emph{regular} buckets, where the ratio of any two prices in each bucket is less than $1+\epst$. Specifically, for each integer $-m\le b<m$, regular bucket $B_b$ includes all items $i\in[n]\setminus(L\cup H)$ such that $(1+\epst)^b \le c_i < (1+\epst)^{b+1}$.
\end{itemize}

We now construct our mechanism by describing the separate mechanism that we hold for the items in each bucket. (As the buyer's valuation is additive, holding a compound mechanism comprised of a number of separate IC and IR mechanisms for pairwise-disjoint sets of items is itself IC and IR.)
\begin{itemize}
\item 
We bundle all of the items in the low bucket $L$ together, and give this bundle to the buyer for free. (Alternatively, we could give any predefined subset of the low bucket to the buyer for free.)
\item
For the high bucket $H$, we offer the buyer the option of purchasing at most one of the items~$i\in H$ (to be chosen by the buyer), for the price $c_i$ of that item.
\item
For a regular bucket $B_b$, we define $\mu_b\eqdef\sum_{i\in B_b}p_B$. We say that $B_b$ is \emph{dense} if $\mu_b>\epst^{-3}$.
\begin{itemize}
\item
If $B_b$ is dense, then we bundle all of the items in the bucket $B_b$ together, and offer this bundle to the buyer for a take-it-or-leave-it price of $(1-\epst)\cdot \mu_b\cdot(1+\epst)^b$.
\item
Otherwise, i.e., if $B_b$ is not dense, we partition the bucket $B_b$ into bundles $B^1_b,B^2_b,\ldots$, s.t.\ for each such bundle $B^j_b$ and for each $i\in B^j_b$, we have that $\sum_{k\in B^j_b\setminus\{i\}}p_k\le\epst$. We offer each such bundle, separately, to the buyer for a take-it-or-leave-it price of $(1+\epst)^b$. (Thus, the mechanism for the items in the bucket $B_b$ is itself comprised of a number of separate mechanisms, one for each bundle.) By \cref{pack-sparse} (see below, after this proof; the \lcnamecref{pack-sparse} is applied after scaling by $\epst$), we have that no more than $\bigl\lceil\epst^{-4}\bigr\rceil$ such bundles are needed when partitioning the bucket~$B_b$.
\end{itemize}
\end{itemize}

Altogether, the number of bundles of items from regular buckets that are offered to the buyer in (separate) take-it-or-leave-it mechanisms is at most
\begin{multline*}
\Bigl|\bigl\{b\in[-m,m-1]\cap\ZZ~\big|~\mu_b>\epst^{-3}\bigr\}\Bigr|+\bigl\lceil\epst^{-4}\bigr\rceil\cdot\Bigl|\bigl\{b\in[-m,m-1]\cap\ZZ~\big|~\mu_b\le\epst^{-3}\bigr\}\Bigr|
\le \\ \le
\bigl\lceil\epst^{-4}\bigr\rceil\cdot2m
=
2\cdot\bigl\lceil\epst^{-4}\bigr\rceil\cdot\bigl\lceil\log_{1+\epst}\nicefrac{n}{\epst}\bigr\rceil
=
2\cdot\bigl\lceil\epst^{-4}\bigr\rceil\cdot\left\lceil\frac{\log_2 n + \log_2(\nicefrac{1}{\epst})}{\log_2(1+\epst)}\right\rceil
\le \\ \le
2\cdot\bigl(\epst^{-4}+1\bigr)\cdot\left(\frac{2\cdot\log_2 n}{\epst}+1\right)
\le
4\cdot\epst^{-4}\cdot\frac{2\cdot\log_2 n}{\epst}
=
\log_2 n\cdot\nicefrac{8}{\epst^5}.
\end{multline*}
Therefore, recalling that $\epst=\nicefrac{\varepsilon}{4}$, the total number of menu entries in the compound mechanism (comprised of the separate mechanisms used for each low, high, or regular bucket, as described above) is less than
\[
1\cdot\bigl(|H|+1\bigr)\cdot2^{\log_2 n\cdot\nicefrac{8}{\epst^5}}
=
\bigl(|H|+1\bigr)\cdot n^{\nicefrac{8}{\epst^5}}
\le
(n+1)\cdot n^{\nicefrac{8}{\epst^5}}
<
n^{\nicefrac{8}{\epst^5}+2}
=
n^{\nicefrac{8192}{\varepsilon^5}+2}
<
n^d.
\]

It remains to analyze the revenue from the compound mechanism, which is, by linearity of expectation, the sum of the revenues from the separate mechanisms held for each of the buckets.

\begin{itemize}
\item 
The revenue from the mechanism used for the low bucket $L$ is obviously $0$. (Note that had each item from $L$ been sold separately, the revenue would have been $\sum_{i\in L}r_i\le\sum_{i \in L}c_i<\sum_{i \in L}\nicefrac{\epst}{n}\le\epst$, and so the loss in revenue is at most an additive~$\epst$.)
\item
We claim that the revenue from the mechanism used for the high bucket $H$ is at least ${(1-\epst)\cdot\sum_{i\in H}r_i}$.
To show this, we note that if the buyer values \emph{exactly} one item $i\in H$ by at least $c_i$, then this item is sold.
As for every $j\in H\setminus\{i\}$ we have that $r_j\le1$ by normalization and that $c_j\ge\nicefrac{n}{\epst}$ by definition of~$H$, we obtain that $p_j=\nicefrac{r_j}{c_j}\ge\nicefrac{\epst}{n}$. Thus, item $i$ is sold with probability at least
\[p_i\cdot\bigl(1-\sum_{j \in H\setminus\{i\}} p_j\bigr) \ge p_i\cdot\left(1-(n-1)\cdot\nicefrac{\epst}{n}\right)>(1-\epst)\cdot p_i.\]
Therefore, the (expected) revenue from the mechanism used for the high bucket is
\[\sum_{i\in H}\prob{v\sim F}{\mbox{item $i$ is sold}}\cdot c_i>\sum_{i\in H}(1-\epst)\cdot p_i\cdot c_i=(1-\epst)\cdot\sum_{i\in H}r_i.\]
(Note that had each item from $H$ been sold separately, the revenue would have been $\sum_{i\in H}r_i$, and so the loss in revenue is at most a multiplicative $(1-\epst)$ factor.)
\item
We claim that the revenue from the mechanism used for a dense regular bucket $B_b$ is at least $(1-\epst)^3\cdot\sum_{i\in B_b}r_i$.
As the price for which we offer the bundle of all of items in this bucket is ${(1-\epst)\cdot \mu_b\cdot(1+\epst)^b}$, it suffices to show that this bundle is sold with probability at least $(1-\epst)$, as this implies that the revenue from this mechanism is at least
\[
(1-\epst)\cdot(1-\epst)\cdot \mu_b\cdot(1+\epst)^b = (1-\epst)^2\cdot\sum_{i\in B_b}p_i \cdot (1+\epst)^b > (1-\epst)^3\cdot\sum_{i\in B_b}p_i \cdot c_i = (1-\epst)^3\cdot\sum_{i\in B_b}r_i.
\]
For every $i\in B_b$, let $X_i$ be an indicator random variable for the buyer valuing item $i$ by at least $c_i$, and so $X_i$ is a Bernoulli
variable taking value 1 with probability $p_i$.  Let $X\eqdef\sum_{i \in b} X_i$ and note that the expectation of $X$ is $\mu_b>\epst^{-3}$. 
Clearly, if $X \ge (1-\epst)\cdot \mu_b$, then the buyer values the bundle by at least $(1-\epst)\cdot \mu_b\cdot\min_{i\in B_b}c_i \ge (1-\epst)\cdot \mu_b\cdot (1+\epst)^b$, and the bundle is sold.
To show that the bundle is sold with probability at least $(1-\epst)$, it therefore suffices to show that $\probl{v\sim F}{X}{(1-\epst)\cdot \mu_b} \le \epst$.
Denoting the standard deviation of $X$ by $\sigma_b$, since $(X_i)_{i\in B_b}$ are independent we have that
$\sigma_b = \sqrt{\sum_{i\in B_b}p_i\cdot(1-p_i)} \le \sqrt{\mu_b}$.
By Chebyshev's inequality, we therefore have that
\[\probl{v\sim F}{X}{(1-\epst)\cdot \mu_b} \le \probl{v\sim F}{X}{\mu_b - (\epst\sqrt{\mu_b})\cdot\sigma_b} \le \frac{1}{\epst^2\cdot\mu_b} < \epst^{-2}\cdot\epst^3 = \epst,\]
as required.
(Note that had each item from $B_b$ been sold separately, the revenue would have been $\sum_{i\in B_b}r_i$, and so the loss in revenue is at most a multiplicative $(1-\epst)^3$ factor,
where the triple loss of $(1-\epst)$ is due once to lower-bounding $c_i$ by $(1+\epst)^b$, a second time to multiplying the requested 
price by $(1-\epst)$ in order to guarantee a sale with high probability, and a third time to the $\epst$ probability of not selling the bundle.)
\item
Finally, we claim that the revenue from the mechanism used for a bundle $B^j_b$ from a nondense regular bucket $B_b$ is at least $(1-\epst)^2\cdot\sum_{i\in B^j_b}r_i$.
As the price that we charge for this bundle is no more than the value of $c_i$ for any item $i\in B^j_b$ in the bundle, we note that, in particular, the bundle is sold if the buyer values some item $i\in B^j_b$ by at least $c_i$. Therefore, the probability that this bundle is sold is at least the probability that the buyer values \emph{exactly} one of the items~$i\in B^j_b$ by at least $c_i$, which, by definition of $B^j_b$, is at least
$\sum_{i \in B^j_b} p_i \cdot \bigl(1-\sum_{k \in B^j_b\setminus\{i\}} p_k\bigr) \ge (1-\epst)\cdot \sum_{i \in B^j_b} p_i$.
Therefore, the revenue from this mechanism is at least
\[
(1-\epst)\cdot\bigl(\sum_{i\in B^j_i}p_i\bigr)\cdot(1+\epst)^b > (1-\epst)^2\cdot\sum_{i\in B^j_b}p_i \cdot c_i = (1-\epst)^2\cdot\sum_{i\in B^j_b}r_i.
\]
(Note that had each item from $B^j_b$ been sold separately, the revenue would have been $\sum_{i\in B^j_b}r_i$, and so the loss in revenue is at most a multiplicative $(1-\epst)^2$ factor,
where the double loss of $(1-\epst)$ is due once to lower-bounding $c_i$ by $(1+\epst)^b$, and a second time to the $\epst$ probability of not selling the bundle.)
\end{itemize}
The total revenue is, therefore, at least
\begin{multline*}
(1-\epst)\cdot\bigl(\sum_{i\in H}r_i\bigr)
+
(1-\epst)^3\cdot\bigl(\smashoperator[r]{\sum_{\substack{i\in B_b :\\ \mu_b>\epst^{-3}}}}r_i\bigr)
+
(1-\epst)^2\cdot\bigl(\smashoperator[r]{\sum_{\substack{i\in B_b :\\ \mu_b\le\epst^{-3}}}}r_i\bigr)
\ge
(1-\epst)^3\cdot\bigl(\sum_{i\in[n]\setminus L}r_i\bigr)\ge \\ \ge
(1-\epst)^3\cdot\bigl(\SRev(F)-\epst\bigr)\ge
(1-\epst)^4\cdot\SRev(F)>
(1-4\cdot\epst)\cdot\SRev(F)=
(1-\varepsilon)\cdot\SRev(F),
\end{multline*}
as required.
\end{proof}

\begin{lemma}\label{pack-sparse}
Let $m\in\RR$, let $p_1,p_2,\ldots,p_m\in\Reals$ be strictly positive numbers, and let $s\eqdef\sum_{i=1}^m p_i$. There exists a partition of $[m]$ into at most $\lceil s\rceil$ sets $B^j$, s.t.\ for every set $B^j$ in the partition, and for every index $i\in B^j$, we have that $\sum_{k\in B^j\setminus\{i\}}p_k\le1$.
\end{lemma}

\begin{proof}
We prove the claim by induction over $\lceil s\rceil$.
We note that the constraint on each set $B^j$ is equivalent to demanding that $\sum_{k\in B^j\setminus\{i\}}p_k\le 1$ for $i\in\arg\min_{k\in B^j} p_k$.

Assume w.l.o.g.\ that $p_1\ge p_2\ge\cdots\ge p_m$ and
let $\ell\eqdef\Max\bigl\{\ell\in\{0,1,\ldots,m-1\}~\big|~{\sum_{k=1}^{\ell}p_k\le1}\bigr\}\in\{0,1,\ldots,m-1\}$.
We define $B^1\eqdef[\ell+1]\subseteq[m]$. By monotonicity of $(p_k)_{k=1}^m$, we have that $i\eqdef\ell+1\in\arg\min_{k\in B^1} p_k$, and by definition of $\ell$, we have that $\sum_{k\in B^1\setminus\{i\}}p_k\le1$. If $\ell=m-1$ (and so $i=m$), then we are done.
Otherwise, by definition, of $\ell$ we have that $\sum_{k\in B^1}p_k=\sum_{k=1}^{\ell+1}p_k>1$ (and hence $\lceil s\rceil>1$), and so $\bigl\lceil\sum_{k\in[m]\setminus B^1}p_k\bigr\rceil\le\lceil s\rceil-1$. Therefore, by the induction hypothesis, the proof is complete as $[m]\setminus B^1$ can be partitioned into at most $\lceil s\rceil-1$ sets $B^j$ for $j\ge2$ s.t.\ for every set $B^j$ in the partition and for every $i\in B^j$, we have that $\sum_{k\in B^j\setminus\{i\}}p_k\le1$.
\end{proof}

We conclude by noting that in \cref{pack-sparse}, the bound of $\lceil s\rceil$ on the number of sets in the partition is in fact tight. This can be demonstrated by setting $m=2\cdot\lceil s\rceil-1$, and setting $p_1\eqdef p_2\eqdef \cdots\eqdef p_{m-1}\eqdef \nicefrac{1}{2}+\varepsilon$ and $p_m\eqdef s-\sum_{k=1}^{m-1}p_k$, for $\varepsilon>0$ sufficiently small s.t.\ $p_m>0$. Indeed, no set $B^j$ that contains two of the indices $1,\ldots,m-1$ may contain any other index without violating the constraint of the \lcnamecref{pack-sparse}.


\section{Lower Bound on\texorpdfstring{\\}{ }Revenue-Approximation Complexity}\label{lower-bound}

In this \lcnamecref{lower-bound}, we prove \cref{cne-lower-bound}, which states that $C(n, \nicefrac{1}{n}) \ge 2^{\Omega(n)}$.
We consider the simple distribution where the value of each item is (independently) either $0$ or $1$, each with probability~$50\%$. Thus, the buyer is interested in (i.e., has value $>0$ for the items in) a random subset of the items (where each subset is chosen with probability exactly~$2^{-n}$), and has value $1$ for each item in this subset.
A revenue-maximizing mechanism for this setting offers each item, separately, for a take-it-or-leave-it price of $1$; viewed as a menu, this mechanism has a menu entry for each nonempty bundle of items $\emptyset\ne S\subseteq[n]$, which sells this bundle $S$ for a price of $|S|$. This mechanism clearly maximizes the revenue, as it extracts the full social welfare (i.e., expected sum of item valuations), which is~$\nicefrac{n}{2}$, as revenue. An alternative mechanism, which has only a single menu entry, sells the whole bundle for a price that is slightly less than
$\nicefrac{n}{2}$, so that with high probability the size of the set of desired items is at least this price (in which case the bundle is sold).
Choosing the price to be $\nicefrac{n}{2}-\omega(\sqrt{n})$ is required in order for the probability of selling the bundle to be subconstant. Therefore, this simple mechanism loses (about) a $\frac{1}{\sqrt{n}}$ fraction of the revenue. In the proof below, we show that such a $\nicefrac{1}{n^c}$ loss of revenue is necessary in every simple-enough mechanism, i.e., in every mechanism having subexponentially many menu entries.

Before proceeding to the proof, we roughly sketch the intuition underlying it. Consider a mechanism that loses less than, say, a $\frac{1}{10n}$ fraction of the revenue.\footnote{Note that since the optimal revenue is $\nicefrac{n}{2}$, a $\frac{1}{10n}$ fraction of the optimal revenue is $\nicefrac{1}{20}$, so we will in fact also show that when all item distributions are supported on $[0,1]$, to guarantee an additive loss of at most a fixed $\nicefrac{1}{20}$, a menu size that is exponential in $n$ is required.} This means that for most of the $2^n$ possible bundles of desired items, the payment
extracted from a buyer interested in (precisely the items in) that bundle must be higher than the bundle size minus $\nicefrac{1}{2}$.  Let us say that a buyer type that is interested in (precisely the items in) a bundle $S$ \emph{pays full price} if the mechanism extracts a payment higher than $|S|-\nicefrac{1}{2}$ from a buyer of this type.
The basic idea is that, roughly speaking, if a buyer type that is interested in a bundle $S$ pays full price, then this buyer type ``should'' have a menu entry ``for itself'' that allocates $S$ for a price close to $|S|$.
The reason is that two buyer types interested in bundles of different sizes cannot both choose the same entry and both pay full price (as the buyer type interested in the smaller-sized bundle would not want to buy at the ``full price'' of the other buyer). We note that it is true that buyer types interested in identically sized bundles $T\ne S$ can still ``share'' the same menu entry and both pay full price; indeed, this can happen if the entry that they both choose allocates the union $S \cup T$ for a price close to $|S|=|T|$. This, however, would imply that a buyer type interested in $S \cup T$ does not pay full price, since this menu entry already offers a discount for the bundle $S\cup T$.  Our proof formalizes this intuition: if ``too many'' buyer types pay full price and choose the same menu entry (which must be the case with a small menu that extracts almost all of the revenue), then we show that buyer types interested in many subsets or their union cannot pay full price, leading to a significant loss in revenue.

\begin{proof}[Proof of \cref{cne-lower-bound}]
We will prove that $C\bigl(n,\frac{1}{10n}\bigr)>2^{\nicefrac{n}{10}}$. We note that this implies the theorem, as $C\bigl(10n,\frac{1}{10n}\bigr)\ge C\bigl(n,\frac{1}{10n}\bigr)$, which can be seen by fixing the value of $9n$ of the $10n$ items to $0$ with probability $1$. Let $n\in\NN$ (throughout the proof, we will add requirements for $n$ to be large enough to accommodate certain conditions).
As stated above, we focus on the product distribution $F\in\Delta(\RR)^n$, where the value of each of the $n$ items is either $0$ or $1$, each independently with probability $50\%$.\footnote{So, we in fact prove the \lcnamecref{cne-lower-bound} by showing that $C\bigl(10n,\frac{1}{10n}\bigr)>2^{\nicefrac{(10n)}{100}}$ by focusing on a product distribution over $10n$ items, where $9n$ of the items have value $0$ with probability $1$, and each of the remaining $n$ items has value either~$0$ or $1$, each independently with probability $50\%$.}
Assume for contradiction that there exists an IC and IR mechanism that has at most $2^{\nicefrac{n}{10}}$ menu entries, but that obtains from $F$ a revenue of at least $\bigl(1-\frac{1}{10n}\bigr)\cdot\Rev(F)=\bigl(1-\frac{1}{10n}\bigr)\cdot\nicefrac{n}{2}$.

For each set of items $S\subseteq[n]$, we identify $S$ with the buyer type valuing each item in $S$ by~$1$ and each item not in $S$ by $0$. We say that a buyer type $S\subseteq[n]$ \emph{pays full price} if the payment extracted from a buyer of this type is strictly higher than $|S|-\nicefrac{1}{2}$. Our strategy is to define a certain distribution over pairs of buyer types (where the set of items associated with the first buyer type is contained in the one associated with the second, i.e., the first buyer type is interested in a subset of the items that interest the second), and reach a contradiction by showing that the probability that both of these buyer types pay full price should on the one hand be large due to the assumption of a slight revenue loss,
but on the other hand be small due to the assumption of a small menu size (recall the informal discussion preceding the proof, roughly explaining why two buyer types that are interested in two sets ordered by inclusion and that are both paying full price must ``usually'' choose distinct menu entries). We begin by defining this distribution over pairs of buyer types.

\paragraph{Construction} We choose two buyer types (sets of items) $S_0, S_1$ at random as follows: we first choose a set of items uniformly at random among all $2^n$ sets, and then
choose a random item $i\in[n]$ and take the symmetric difference of the first set with $\{i\}$ to obtain a second set.
Thus, exactly one of the two chosen sets contains the element $i$, while the two sets are identical in their containment of all other items.
Let us denote the set that does not contain $i$ by $S_0$, and the set that contains $i$ by~$S_1$.
We note that while the distributions of each of $S_0$ and $S_1$ are not uniform (e.g., the expected size of~$S_0$ is $\frac{n-1}{2}$ and of $S_1$ is $\frac{n+1}{2}$, while the expected size of a uniformly chosen set is $\nicefrac{n}{2}$), it is nonetheless the case that the probability that any of these distributions gives to any event is at most twice the probability that the uniform distribution gives to the same event (since the average of the probabilities given to this event by the distributions of $S_0$ and of $S_1$ is exactly that given by the uniform distribution).

\paragraph{Lower Bound}
We lower-bound the probability that both buyer types $S_0$ and $S_1$ pay full price by using the union bound.
We first note that at most $\nicefrac{1}{10}$ of the buyer types $S\subseteq[n]$ do \emph{not} pay full price; otherwise, the revenue loss due to all such buyer types would have been higher than an additive $\frac{1}{20}$, i.e., higher than a $\frac{1}{10n}$ fraction of the maximal revenue, contradicting the assumption that the given mechanism loses at most a $\frac{1}{10n}$ fraction of the maximal revenue. Next, notice that before the labeling, each of $S_0$ and $S_1$ was uniformly distributed
among all possible $2^n$ buyer types. Therefore, for each one of the (unlabeled) buyer types, the probability that it does not pay full price is at most~$\nicefrac{1}{10}$. Hence, taking the union bound, the probability that at least one of the (unlabeled) buyer types does not pay full price is at most~$\nicefrac{1}{5}$, and so the probability that both pay full price is at least~$\nicefrac{4}{5}$. Following the strategy outlined above, we will reach a contradiction by showing that this probability need in fact be less than $\nicefrac{4}{5}$ when taking into account the assumption that the given mechanism has a small menu size.

\paragraph{Upper Bound}
The derivation of an upper bound on the probability that both buyer types $S_0$ and $S_1$ pay full price is more subtle. Consider the following way of obtaining the exact same distribution
for~$(S_0,S_1)$: first, $S_0$ is chosen according to the marginal distribution on it, and then $S_1$ is obtained by adding a random item $i \not\in S_0$ (uniformly chosen from all items not in $S_0$) to $S_0$.

In our derivation of the upper bound, we condition upon whether the menu entry chosen by $S_0$ is chosen by ``many'' buyer types that pay full price.
Formally, let us say that a menu entry from the given mechanism is \emph{tiny} if at most $2^{0.8n}$ buyer types that pay full price choose this entry. Our upper bound on the probability that both buyer types $S_0$ and $S_1$ pay full price will be the sum of the probability that $S_0$ pays full price and chooses a tiny menu entry (which is obviously not smaller than the probability that $S_0$ and $S_1$ both pay full price, and $S_0$ chooses a tiny menu entry) and the probability that $S_1$ pays full price, conditioned upon $S_0$ paying full price and not choosing a tiny menu entry (which is obviously not smaller than the probability that $S_0$ and $S_1$ both pay full price, and $S_0$ does not choose a tiny menu entry). We now estimate each of these two probabilities.

Since there are at most $2^{\nicefrac{n}{10}}$ menu entries in the given mechanism, the total number of buyer types that pay full price and choose a tiny menu entry is at most $2^{0.9n}$.
Thus, the probability that $S_0$ pays full price and chooses a tiny menu entry is at most $2\cdot\frac{2^{0.9n}}{2^n}$ (by the above observation that the probability of any event for $S_0$ is at most twice the probability of that event for a uniform set), i.e., this probability decreases exponentially with $n$.

Following the strategy outlined above, we proceed to upper-bound the probability that $S_1$ pays full price, conditioned upon $S_0$ paying full price and not choosing a tiny menu entry. Assume, therefore, that $S_0$ pays full price but chooses some nontiny menu entry $e$. We consider the probabilities with which $e$ allocates each of the items. By IR, every item that is contained in every set that (as a buyer type) pays full price and chooses $e$ must be allocated by $e$ with probability strictly greater than~$\nicefrac{1}{2}$ (since if, for some buyer type $S\subseteq[n]$, even a single item $i\in S$ is allocated by $e$ with probability at most~$\nicefrac{1}{2}$, then $S$ would receive negative utility from paying strictly more than $|S|-\nicefrac{1}{2}$ for~$e$).
Let $U$ be the set of items that are allocated with probability strictly greater than $\nicefrac{1}{2}$ by $e$.
Since $e$ is not tiny, there are more than $2^{0.8n}$ buyer types $S$ that pay full price and choose it, and so $|U|>0.8n$ (since the total number of
subsets of $U$ is at most $2^{|U|}$.)  

We now come to the crux of our argument.  Consider a set $T$ that is contained in $U$, but is
a proper \emph{superset} of one of the sets  (buyer types) $S\subseteq[n]$ that pay full price and choose $e$. By IC, we claim that $T$ cannot be paying full price in the given mechanism.
Indeed, $e$ already offers buyer type~$T$ utility strictly greater than $\nicefrac{1}{2}$ (since the utility of $T$ from $e$ is greater by strictly more than~$\nicefrac{1}{2}$ than the utility of buyer type $S$ from $e$, which by IR is nonnegative), and a buyer that pays full price has, by definition, utility strictly less than $\nicefrac{1}{2}$.

Since $S_1$ is a strict superset of $S_0$, and since $S_0$ chooses $e$ and pays full price, we therefore have that if $S_1\subseteq U$, then $S_1$ does not pay full price. Therefore, the probability that $S_1$ is not contained in $U$ is an upper bound on the probability that $S_1$ pays full price (when both are conditioned upon $S_0$ paying full price and not choosing a tiny menu entry).
Recall that in our ``second way'' of obtaining the distribution for $(S_0,S_1$), the set $S_1$ is obtained by adding a random item $i \not\in S_0$ (uniformly chosen from all items not in $S_0$) to $S_0$.  For any $S_0$, the probability that $i \not\in U$ is thus at most $\frac{n-|U|}{n-|S_0|}<\frac{0.2n}{n-|S_0|}$, and since 
the probability that $|S_0| > 0.6n$ (i.e., that $\frac{0.2n}{n-|S_0|}>\nicefrac{1}{2}$) decreases exponentially with $n$ (even after conditioning upon $S_0$ paying full price and not choosing a tiny menu entry\footnote{Indeed, recall that at least $\nicefrac{9}{10}$ of the buyer types pay full price; therefore, the probability that $S_0$ pays full price is at least $\nicefrac{8}{10}$ (by the above observation that the probability of any event for $S_0$ is at most twice the probability of that event for a uniform set). Recall also that the probability that $S_0$ pays full price and chooses a tiny menu entry decreases exponentially with $n$; therefore, for large enough $n$, the probability that $S_0$ pays full price and does not choose a tiny menu entry is at least $\nicefrac{7}{10}$.}), 
we get that conditioned upon $S_0$ paying full price and not choosing a tiny menu entry, the probability that $S_1$ pays full price is, for large enough $n$, at most, say, $\nicefrac{3}{5}$.

\paragraph{Contradiction}
Following the strategy outlined above, we can now upper-bound the probability that both $S_0$ and~$S_1$ pay full price, by summing the probability that $S_0$ pays full price and chooses
a tiny menu entry (which decreases exponentially with $n$) and the probability that $S_1$ pays full price conditioned upon $S_0$ paying full price and
not choosing a tiny menu entry (which we upper-bounded by $\nicefrac{3}{5}$ for large enough $n$). Summing these two probabilities, for large enough $n$, we get strictly less than $\nicefrac{4}{5}$, which is the 
lower bound computed above for the probability of the same event --- a contradiction.
\end{proof}

\section*{Acknowledgments}
\begin{sloppypar}
Yannai Gonczarowski was supported in part by the Adams Fellowship Program of the Israel Academy of Sciences and Humanities. The work of Yannai Gonczarowski was supported in part by the European Research Council under the European Community's Seventh Framework Programme (FP7/2007-2013) / ERC grant agreement no.\ [249159]. 
The work of Noam Nisan was supported by ISF grant 1435/14 administered by the Israeli Academy of Sciences, by
Israel-USA Bi-national Science Foundation (BSF) grant number 2014389, and by the European Research Council (ERC) under the European Union's
Horizon 2020 research and innovation programme (grant agreement No 740282).
We thank the anonymous referees for helpful feedback.
\end{sloppypar}

\bibliographystyle{abbrvnat}
\bibliography{bib}

\begin{thebibliography}{30}
\providecommand{\natexlab}[1]{#1}
\providecommand{\url}[1]{\texttt{#1}}
\expandafter\ifx\csname urlstyle\endcsname\relax
  \providecommand{\doi}[1]{doi: #1}\else
  \providecommand{\doi}{doi: \begingroup \urlstyle{rm}\Url}\fi

\bibitem[Babaioff et~al.(2014)Babaioff, Immorlica, Lucier, and Weinberg]{BILW}
M.~Babaioff, N.~Immorlica, B.~Lucier, and S.~M. Weinberg.
\newblock A simple and approximately optimal mechanism for an additive buyer.
\newblock In \emph{Proceedings of the IEEE 55th Annual Symposium on Foundations
  of Computer Science (FOCS)}, pages 21--30, 2014.

\bibitem[Babaioff et~al.(2017)Babaioff, Gonczarowski, and Nisan]{bgn17}
M.~Babaioff, Y.~A. Gonczarowski, and N.~Nisan.
\newblock The menu-size complexity of revenue approximation.
\newblock In \emph{Proceedings of the 49th Annual ACM Symposium on Theory of
  Computing (STOC)}, pages 869--877, 2017.

\bibitem[Babaioff et~al.(2018)Babaioff, Nisan, and Rubinstein]{bnr18}
M.~Babaioff, N.~Nisan, and A.~Rubinstein.
\newblock Optimal deterministic mechanisms for an additive buyer.
\newblock In \emph{Proceedings of the 19th ACM Conference on Economics and
  Computation (EC)}, page 429, 2018.

\bibitem[Briest et~al.(2010)Briest, Chawla, Kleinberg, and Weinberg]{pricerand}
P.~Briest, S.~Chawla, R.~Kleinberg, and S.~M. Weinberg.
\newblock Pricing randomized allocations.
\newblock In \emph{Proceedings of the 21st Annual ACM-SIAM Symposium on
  Discrete Algorithms (SODA)}, pages 585--597, 2010.

\bibitem[Cai and Zhao(2017)]{CZ17}
Y.~Cai and M.~Zhao.
\newblock Simple mechanisms for subadditive buyers via duality.
\newblock In \emph{Proceedings of the 49th Annual ACM Symposium on Theory of
  Computing (STOC)}, pages 170--183, 2017.

\bibitem[Chawla et~al.(2010)Chawla, Malec, and Sivan]{chawla10b}
S.~Chawla, D.~L. Malec, and B.~Sivan.
\newblock The power of randomness in {B}ayesian optimal mechanism design.
\newblock In \emph{Proceedings of the 11th ACM Conference on Electronic
  Commerce (EC)}, pages 149--158, 2010.

\bibitem[Daskalakis and Weinberg(2012)]{DW12}
C.~Daskalakis and S.~M. Weinberg.
\newblock Symmetries and optimal multi-dimensional mechanism design.
\newblock In \emph{Proceedings of the 13th ACM Conference on Electronic
  Commerce (EC)}, pages 370--387, 2012.

\bibitem[Daskalakis et~al.(2013)Daskalakis, Deckelbaum, and Tzamos]{DDT13}
C.~Daskalakis, A.~Deckelbaum, and C.~Tzamos.
\newblock Mechanism design via optimal transport.
\newblock In \emph{Proceedings of the Fourteenth ACM Conference on Electronic
  Commerce (EC)}, pages 269--286, 2013.

\bibitem[Devanur and Weinberg(2017)]{dw17}
N.~R. Devanur and S.~M. Weinberg.
\newblock The optimal mechanism for selling to a budget constrained buyer: The
  general case.
\newblock In \emph{Proceedings of the 18th ACM Conference on Economics and
  Computation (EC)}, pages 39--40, 2017.

\bibitem[Devanur et~al.(2020)Devanur, Goldner, Saxena, Schvartzman, and
  Weinberg]{dgssw18}
N.~R. Devanur, K.~Goldner, R.~R. Saxena, A.~Schvartzman, and S.~M. Weinberg.
\newblock Optimal mechanism design for single-minded agents.
\newblock Mimeo, 2020.

\bibitem[Dughmi et~al.(2014)Dughmi, Han, and Nisan]{DLN14}
S.~Dughmi, L.~Han, and N.~Nisan.
\newblock Sampling and representation complexity of revenue maximization.
\newblock In \emph{Proceedings of the 10th Conference on Web and Internet
  Economics (WINE)}, pages 277--291, 2014.

\bibitem[Fiat et~al.(2016)Fiat, Goldner, Karlin, and Koutsoupias]{fgkk16}
A.~Fiat, K.~Goldner, A.~R. Karlin, and E.~Koutsoupias.
\newblock The {F}ed{E}x problem.
\newblock In \emph{Proceedings of the 17th ACM Conference on Economics and
  Computation (EC)}, pages 21--22, 2016.

\bibitem[Giannakopoulos and Koutsoupias(2014)]{GK14}
Y.~Giannakopoulos and E.~Koutsoupias.
\newblock Duality and optimality of auctions for uniform distributions.
\newblock In \emph{Proceedings of the Fifteenth ACM Conference on Economics and
  Computation (EC)}, pages 259--276, 2014.

\bibitem[Giannakopoulos and Koutsoupias(2015)]{GK15}
Y.~Giannakopoulos and E.~Koutsoupias.
\newblock Selling two goods optimally.
\newblock In \emph{Proceedings of the 42nd International Colloquium on
  Automata, Languages, and Programming (ICALP)}, pages 650--662, 2015.

\bibitem[Goldner and Gonczarowski(2018)]{gg18-tutorial}
K.~Goldner and Y.~A. Gonczarowski.
\newblock The menu size of precise and approximate revenue-maximizing auctions.
\newblock Tutorial, The 19th ACM Conference on Economics and Computation (EC),
  2018.
\newblock URL
  \url{http://yannai.gonch.name/scientific/ec18-menu-size-tutorial/}.

\bibitem[Gonczarowski(2018)]{g18}
Y.~A. Gonczarowski.
\newblock Bounding the menu-size of approximately optimal auctions via
  optimal-transport duality.
\newblock In \emph{Proceedings of the 50th Annual ACM Symposium on Theory of
  Computing (STOC)}, pages 123--131, 2018.

\bibitem[Hart and Nisan(2012)]{hart-nisan-a}
S.~Hart and N.~Nisan.
\newblock Approximate revenue maximization with multiple items.
\newblock In \emph{Proceedings of the 13th ACM Conference on Electronic
  Commerce (EC)}, page 656, 2012.

\bibitem[Hart and Nisan(2013)]{hart-nisan-b}
S.~Hart and N.~Nisan.
\newblock The menu-size complexity of auctions.
\newblock In \emph{Proceedings of the Fourteenth ACM Conference on Electronic
  Commerce (EC)}, page 565, 2013.

\bibitem[Hart and Nisan(2014)]{hart-nisan-combined}
S.~Hart and N.~Nisan.
\newblock How good are simple mechanisms for selling multiple goods?
\newblock Discussion Paper 666, Center for the Study of Rationality, Hebrew
  University of Jerusalem, 2014.

\bibitem[Hart and Reny(2015)]{HR15}
S.~Hart and P.~J. Reny.
\newblock Maximal revenue with multiple goods: Nonmonotonicity and other
  observations.
\newblock \emph{Theoretical Economics}, 10\penalty0 (3):\penalty0 893--922,
  2015.

\bibitem[Kothari et~al.(2019)Kothari, Mohan, Schvartzman, Singla, and
  Weinberg]{KMSSW19}
P.~Kothari, D.~Mohan, A.~Schvartzman, S.~Singla, and S.~M. Weinberg.
\newblock Approximation schemes for a unit-demand buyer with independent items
  via symmetries.
\newblock In \emph{Proceedings of the IEEE 60th Annual Symposium on Foundations
  of Computer Science (FOCS)}, pages 220--232, 2019.

\bibitem[Li and Yao(2013)]{LY13}
X.~Li and A.~C.-C. Yao.
\newblock On revenue maximization for selling multiple independently
  distributed items.
\newblock \emph{Proceedings of the National Academy of Sciences (PNAS)},
  110\penalty0 (28):\penalty0 11232--11237, 2013.

\bibitem[Manelli and Vincent(2006)]{MV06}
A.~M. Manelli and D.~R. Vincent.
\newblock Bundling as an optimal selling mechanism for a multiple-good
  monopolist.
\newblock \emph{Journal of Economic Theory}, 127\penalty0 (1):\penalty0 1--35,
  2006.

\bibitem[McAfee and McMillan(1988)]{MM88}
R.~P. McAfee and J.~McMillan.
\newblock Multidimensional incentive compatibility and mechanism design.
\newblock \emph{Journal of Economic Theory}, 46\penalty0 (2):\penalty0
  335--354, 1988.

\bibitem[Morgenstern and Roughgarden(2015)]{MR15}
J.~Morgenstern and T.~Roughgarden.
\newblock On the pseudo-dimension of nearly optimal auctions.
\newblock In \emph{Proceedings of Advances in Neural Information Processing
  Systems 28 (NIPS)}, pages 136--144, 2015.

\bibitem[Myerson(1981)]{Myerson81}
R.~Myerson.
\newblock Optimal auction design.
\newblock \emph{Mathematics of Operations Research}, 6\penalty0 (1):\penalty0
  58--73, 1981.

\bibitem[Rubinstein and Weinberg(2015)]{RW15}
A.~Rubinstein and S.~M. Weinberg.
\newblock Simple mechanisms for a subadditive buyer and applications to revenue
  monotonicity.
\newblock In \emph{Proceedings of the Sixteenth ACM Conference on Economics and
  Computation (EC)}, pages 377--394, 2015.

\bibitem[Saxena et~al.(2018)Saxena, Schvartzman, and Weinberg]{ssw18}
R.~R. Saxena, A.~Schvartzman, and S.~M. Weinberg.
\newblock The menu complexity of ``one-and-a-half-dimensional'' mechanism
  design.
\newblock In \emph{Proceedings of the 29th Annual ACM-SIAM Symposium on
  Discrete Algorithms (SODA)}, pages 2026--2035, 2018.

\bibitem[Thanassoulis(2004)]{Tha04}
J.~Thanassoulis.
\newblock Haggling over substitutes.
\newblock \emph{Journal of Economic Theory}, 117\penalty0 (2):\penalty0
  217--245, 2004.

\bibitem[Yao(2015)]{Y15}
A.~C.-C. Yao.
\newblock An $n$-to-$1$ bidder reduction for multi-item auctions and its
  applications.
\newblock In \emph{Proceedings of the Twenty-Sixth Annual ACM-SIAM Symposium on
  Discrete Algorithms (SODA)}, pages 92--109, 2015.

\end{thebibliography}

\appendix

\section{Menu Size and Communication Complexity}\label{comm}

One may view the menu size of a mechanism as exactly characterizing its (deterministic) communication complexity.  To see this easy correspondence, we should first carefully define the communication complexity.  In the communication complexity scenario, we have a commonly known mechanism (in the intended application it may be determined by the commonly known distributions $F_1,\ldots,F_n$ on item values). Our single buyer also has a privately known valuation, and the seller has no private knowledge.  The mechanism is just a function that maps the buyer's valuation to the allocation and payment for this valuation.

\begin{definition}[Deterministic Communication Complexity]
The deterministic communication complexity of a given mechanism is the minimal number of bits that must be exchanged, in the worst case, by the buyer and seller in order for both of them to compute the allocation and payment for the buyer's valuation.
\end{definition}

We can immediately notice that since the seller has no private information, then the buyer may simulate him completely, and thus there is never any need for the seller to transmit anything.  It follows that without loss of generality, the communication is one-sided: the buyer transmits some information about his valuation, which suffices to determine the outcome of the mechanism (allocation and payment).  Thus, if some (prefix free) communication protocol for the mechanism uses at most $c$ bits of communication, then there are at most $2^c$ possible outcomes, i.e., the menu size of the mechanism is at most $2^c$. Conversely, if the mechanism has a menu size of $C$, then clearly the buyer need only transmit her chosen menu entry, which requires $\lceil \log_2 C \rceil$ bits of information.  We have thus observed the following.

\begin{proposition}\label{cc}
The deterministic communication complexity of a mechanism is exactly the logarithm (base 2, rounded up) of the menu size of the mechanism.
\end{proposition}

Note also that once the mechanism is incentive compatible, this characterization of the communication complexity holds whether or not the communication protocol is required to be incentive compatible: the lower bound on the communication complexity applies even without any incentive requirements, while the upper bound protocol is clearly incentive compatible.

The definition of randomized communication complexity may be more subtle: one definition would just require a randomized protocol that computes, with high probability, the required outcome.  A more natural definition, however, would just require the distribution of the outcomes of the protocol to be that specified by the mechanism.

\begin{definition}[Randomized Communication Complexity]
The randomized communication complexity of a mechanism is the expected number of bits of communication required, for the worst case input, in a  randomized protocol whose probability of allocating each item and whose expected payment are as specified by the mechanism.
\end{definition}

As it turns out, this definition of randomized communication complexity is a much easier benchmark to satisfy, at least as long as we allow ``public coins'' (i.e., as long as we assume that the buyer and seller share a randomly generated string, whose length is not counted toward the communication).  For instance, one may in fact implement any single-item mechanism with a single bit of communication under a public-coin model. By a theorem of \cite{Myerson81}, an IC single-item mechanism is completely specified by\ \ 1) the probability of allocation $x(v)$ for every value $v$, where $x(\cdot)$ is a nondecreasing function, and\ \ 2) the payment $p_0$ of a buyer with valuation $0$. Suppose that the buyer and seller choose (jointly, without any communication cost) a price $p$ at random so that $x(p)$ is exactly uniform on $[0,1]$.  To get the correct allocation, it suffices for the buyer to get the item if
$v \ge p$ (requiring one bit of communication), which indeed happens with probability exactly $x(v)$.  If we charge the buyer exactly $p+p_0$ when she gets the item and $p_0$ when she does not get the item, then the expected payment turns out to be as required.

It is not completely clear to us how large the public-coin randomized complexity can be in the case of multiple items, nor is it clear how large the private-coin randomized complexity can be.

\section{Technical Notes}

\subsection{Infinite Menus}\label{closure}

\subsubsection{Utility-Maximizing Entries}\label{closure-max-utility}

As is well known in the literature, by the taxation principle every (single-buyer) IC mechanism with finitely many possible outcomes can be identified with a finite menu of possible choices for the buyer (where by IC the buyer chooses an entry that maximizes her utility), and vice versa. For IC mechanisms with infinitely many possible outcomes, while it is still true that each such mechanism can be identified with an (infinite) menu of possible choices for the buyer, it is no longer the case that every such menu defines some IC mechanism. Indeed, in a general infinite menu, a utility-maximizing entry for each buyer type does not necessarily exist, yet in menus corresponding to IC mechanisms, such an entry always exists. We note that one way to make sure that a (possibly infinite) menu that we construct indeed defines an IC mechanism is to make sure that this menu is closed (as a subset of $[0,1]^n\times\RR$). (This is the technique employed in our proofs of \cref{expensive-exc,cheap-disc}.\footnote{The reasoning in our proof of \cref{expensive-trim} is more delicate and is detailed in that proof.}) Indeed, for a buyer type $v=(v_1,\ldots,v_n)\in\RRn$, all nonnegative-utility entries from a closed menu $\mech$ lie in the compact set $\mech\cap\Bigl([0,1]^n\times\bigl[0,\sum_{i=1}^n v_i\bigr]\Bigr)$, and by continuity of the utility function, this function attains a maximum value in this compact set.

\subsubsection{Tie-Breaking by Prices}\label{closure-max-price}

\crefpart{revenue}{mech}, which is the standard definition of the revenue obtainable by an IC mechanism, specifies that ties (in utility) between menu entries are broken in favor of higher prices. We note that even if, in some menu, a utility-maximizing menu entry exists for some buyer type (or for all buyer types), then it is not guaranteed that a utility-maximizing entry \emph{with maximal price} (among all utility-maximizing entries) exists for this buyer type. (I.e., it is not guaranteed that the supremum price over all utility-maximizing entries for this buyer type is attained as a maximum.) Indeed, to be completely general, a more subtle definition of the revenue obtainable by an IC mechanism would have been needed (taking, roughly speaking, the supremum revenue over all tie-breaking rules), and even under such a definition, the menu entry of choice of a specific buyer type would not have been well defined (or would have become a limit of menu entries), making reasoning about such mechanisms quite cumbersome. Nonetheless, similarly to above, if an IC mechanism has a closed menu, then it possesses, for each buyer type, a utility-maximizing entry with maximal price, and so the revenue is well defined by \crefpart{revenue}{mech}, without the need for a more subtle definition. Indeed, the set of utility-maximizing entries from a given closed menu for a given buyer type is (by the reasoning given in \cref{closure-max-utility}) a compact set; therefore, it contains an entry with maximum price. The only point in the proofs in this paper where we do not explicitly construct a mechanism by specifying its menu is in the proof of \cref{cne-upper-bound}, where we start with an IC (and IR) mechanism (possibly of infinite size) that obtains revenue close to $\Rev(F)$; to justify the fact that we can assume w.l.o.g.\ that the menu of such a mechanism is closed, we note that if some menu entry $e$ is weakly preferred by some buyer type $v$ to all menu entries in some menu~$\mech$, then by continuity of the utility function, $e$ is weakly preferred by~$v$ also to all menu entries in $\closure{\mech}$ --- the \emph{closure} of $\mech$ (in $[0,1]^n\times\RR$). Therefore, given an IC mechanism~$\mech$ that obtains revenue $R$ from a distribution $F$ under some tie-breaking rule, we have that $\closure{\mech}$, under price-maximization tie-breaking (which is well defined since $\closure{\mech}$ is closed), obtains revenue at least~$R$ from $F$.

\subsection{Arbitrary Tie-Breaking}\label{tie-breaking}

As discussed above, the (standard) definition that we use for the revenue obtainable by a given IC mechanism (\crefpart{revenue}{mech}) depends on tie-breaking being performed in favor of higher prices. Nonetheless, we emphasize that the definition of the revenue obtainable from a given distribution $F$, whether constrained by the menu size (i.e., $\Rev_C$) or unconstrained (i.e., $\Rev$), does not depend on the tie-breaking rule. Indeed, if a mechanism $\mech$ obtains revenue $R$ w.r.t.\ tie-breaking in favor of high prices, then for arbitrarily small $\varepsilon>0$, multiplying the price of each menu entry in $\mech$ by $(1-\varepsilon)$ and taking the closure of the resulting menu (equivalently, multiplying the price of each menu entry in $\closure{\mech}$ by $(1-\varepsilon)$), yields a mechanism with the same menu size (since a finite menu is always closed) that obtains revenue at least $(1-\varepsilon)\cdot R$ w.r.t.\ \emph{any} tie-breaking rule (since multiplying each price by $(1-\varepsilon)$ breaks ties in favor of higher-priced menu entries, and can only cause a buyer type to ``jump'' to even higher-costing menu entries).
As the definitions of $\Rev_C$ and $\Rev$ therefore do not depend on the tie-breaking rule, our results hold for any tie-breaking rules, e.g., even tie-breaking in favor of low prices in the definition of $\Rev_C$ and in favor of high prices in the definition of $\Rev$.

\section{A Tighter Version of Lemma~\ref{cheap-disc}}\label{logn}

\begin{lemma}\label{cheap-disc-log}
Let $n\in\NN$, let $\thresh\in\RR$, and let $\varepsilon\in(0,1)$. For every $F\in\Delta(\RRn)$ and for every IC and IR $n$-item mechanism $\mech$, there exists an IC and IR $n$-item mechanism $\mech'$ such that both of the following hold.
\begin{itemize}
\item
$\Rev_{\mech'}(F) > (1-\varepsilon)\cdot\Rev_{\mech}(F)-\varepsilon$.
\item
There are as many menu entries that cost more than $(1-\nicefrac{\varepsilon}{3})^2\cdot\thresh$ in $\mech'$ as there are that cost more than $\thresh$ in~$\mech$.
\item
There are fewer than $(\nicefrac{\log\thresh}{\varepsilon})^{O(n)}$ entries in $\mech'$ that cost at most $(1-\nicefrac{\varepsilon}{3})^2\cdot\thresh$.
\end{itemize}
\end{lemma}

\begin{proof}
The proof uses discretization techniques from \cite{hart-nisan-b} \cite[see also][]{DLN14}. 
Similarly to the proof of \cref{cheap-disc}, we will gradually define the mechanism $\mech'$ through interim mechanisms, however we will have more than one such interim mechanism. We will first define an interim mechanism $\mech^{(2)}$ that will have prices (of cheap menu entries) discretized, and we will then use it to sequentially define two more interim mechanisms~$\mech^{(3)}$ and~$\mech^{(4)}$ that will also have allocations discretized. Finally, as in the proof of \cref{cheap-disc}, we will derive the required mechanism~$\mech'$ by defining it to be the mechanism offering only those menu entries of $\mech^{(4)}$ that are in fact chosen by any buyer type. Set $\epst\eqdef\nicefrac{\varepsilon}{3}$. Let $K\eqdef\Bigl\lceil\frac{\log\nicefrac{\thresh}{\varepsilon}}{\log(1+\epst^2)}\Bigr\rceil=O(\nicefrac{1}{\epst^2}\cdot\log\nicefrac{\thresh}{\varepsilon})$ and let $\pi=\sqrt[K]{\nicefrac{\thresh}{\varepsilon}}$. Note that $K$ is the smallest natural number such that $\pi\le(1+\epst^2)$.

We start by defining $\mech^{(2)}$. Each of the cheap menu entries of this mechanism, except for $(\vec{0};0)$, will only have one of the $K+1$ prices $p_{\ell}\eqdef(1-\varepsilon)\cdot\pi^\ell\cdot\varepsilon$ for $\ell\in\{0,\ldots,K\}$. (Note that $p_K=(1-\varepsilon)\cdot\thresh$.) We construct $\mech^{(2)}$ as follows:
\begin{itemize}
\item
For every menu entry $e=(\vec{x};p)\in\mech$ with $p\in[\varepsilon,\thresh]$, we define $p'$ to be $(1-\epst)\cdot p$ rounded up to the nearest price in the set $\{p_0,\ldots,p_K\}$,\footnote{Note that indeed $(1-\epst)\cdot p\le(1-\epst)\cdot\thresh=p_K$.} and add the menu entry $e'=(\vec{x};p')$ to $\mech^{(2)}$. (Note that $p'\le p_K=(1-\epst)\cdot\thresh$.)
\item
For every menu entry $e=(\vec{x};p)\in\mech$ with with $p>\thresh$, as well as for the menu entry $e=(\vec{x};p)=(\vec{0};0)\in\mech$, we define $p'=(1-\epst)\cdot p$ (but without any rounding), and add the menu entry $e'=(\vec{x};p')$ to $\mech^{(2)}$. (In particular, $\mech^{(2)}$ is IR.)
\item
We ignore any other menu entries of $\mech$ (i.e., any menu entries other than $(\vec{0};0)$. whose price is less than $\varepsilon$).
\end{itemize}
We note that for each entry $e'=(\vec{x};p')$ that is added to $\mech^{(2)}$, the price of its counterpart $e=(\vec{x};p)$ from $\mech$ satisfies $(1-\epst)\cdot p\le p' < \pi\cdot(1-\epst)\cdot p\le(1+\epst^2)\cdot(1-\epst)\cdot p$.
Finally, as in the proof of \cref{cheap-disc}, we define $\mech^{(2)}$ to be the closure of the set of menu entries added above to~$\mech^{(2)}$.

We note that $\Rev_{\mech^{(2)}}(F) \ge (1-\epst)^2\cdot(\Rev_{\mech}(F)-\varepsilon)$. The $\varepsilon$ additive loss is due to possibly losing all revenue from buyers who pay less than $\varepsilon$ in $\mech$. The multiplicative bound on the loss is ensured by the fact that for any two menu entries $e=(\vec{x};p)$ and $f=(\vec{y};q)$ from $\mech$, if $q<(1-\epst)\cdot p$, then $p'-q'<(1+\epst^2)\cdot(1-\epst)\cdot p-(1-\epst)\cdot q<p-q-\epst^3\cdot p<p-q$ (and recall that the allocations of $f$ and $e$ are unchanged), so any buyer who chooses $e$ from $\mech$, prefers $e'$ over any menu entry from $\mech^{(2)}$ that costs less than $(1-\epst)\cdot p'$, and so pays at least $(1-\epst)\cdot p'\ge(1-\epst)^2\cdot p$ in $\mech^{(2)}$.

\begin{sloppypar}
We now define $\mech^{(3)}$ by taking $\mech^{(2)}$ and multiplying every allocation probability and every price by $(1-\epst)$. Therefore, $\mech^{(3)}$ is IR, has revenue $\Rev_{\mech^{(3)}}(F) ={(1-\epst)\cdot\Rev_{\mech^{(2)}}(F)}\ge (1-\epst)^3\cdot(\Rev_{\mech}(F)-\varepsilon)>(1-\varepsilon)\cdot\Rev_{\mech}(F)-\varepsilon$, its menu entries that cost more than $(1-\epst^2)\cdot\thresh$ are in one-to-one correspondence with the entries of $\mech$ that costs more than $\thresh$, and its other menu entries apart from $(\vec{0};0)$ only have prices in $\{(1-\epst)\cdot p_1,(1-\epst)\cdot p_2,\ldots,(1-\epst)\cdot p_K\}$. Also, each allocation probability of~$\mech^{(3)}$ is at most $1-\epst$.
\end{sloppypar}

Finally, we define $\mech^{(4)}$. Set $\chi\eqdef\frac{\epst}{K+1}$. We construct $\mech^{(4)}$ as follows:
\begin{itemize}
\item
We add the menu entry $e=(\vec{0};0)\in\mech^{(3)}$ to $\mech^{(4)}$, unchanged.
\item
For every other menu entry $e=(\vec{x};p)\in\mech^{(3)}$ with $p\le(1-\epst^2)\cdot\thresh$, recall that $p=(1-\epst)\cdot p_\ell$ for some $\ell\in\{0,1,\ldots,K\}$. For each $i\in[n]$, we define $x'_i=\lfloor x_i \rfloor_{\chi}+\ell\cdot\chi$, i.e., $x_i$ rounded down to the nearest multiple of $\chi$, plus a ``bonus'' probability that increases with $p_\ell$ (note that $x'_i\le1$ since $x_i\le1-\epst$), and add the menu entry $(\vec{x}';p)$ to $\mech^{(4)}$.
\item
For every menu entry $e=(\vec{x};p)\in\mech^{(3)}$ with $p>(1-\epst)^2\cdot\thresh$, we add that menu entry to $\mech^{(4)}$, with each allocation probability increased by $\epst$. (These are a well defined probabilities since in $\mech^{(3)}$ all allocation probabilities are at most $1-\epst$).
\end{itemize}
Finally, we define $\mech^{(4)}$ to be the closure of the set of menu entries added above to~$\mech^{(4)}$. We note that all allocation probabilities have been weakly increased from $\mech^{(3)}$ to $\mech^{(4)}$, and for every two menu entries $e,f\in\mech^{(3)}$, if the price of $e$ is higher than that of~$f$, then each allocation probability in $e$ was weakly increased (even when taking into account also the rounding) by no less than the corresponding allocation probability in $f$, and therefore the revenue from no buyer type decreases, and so $\Rev_{\mech^{(4)}}(F)\ge\Rev_{\mech^{(3)}}(F)>(1-\varepsilon)\cdot\Rev_{\mech}(F)-\varepsilon$.

Finally, we will define $\mech'$ to be the mechanism offering only the menu entries from $\mech^{(4)}$ that are in fact chosen by any buyer type. We conclude the proof by noting that, as required, the number of menu entries (not including $(\vec{0};0)$) that cost at most $(1-\epst)^2\cdot\thresh$ in~$\mech'$ is at most $(K+1)\cdot(\nicefrac{1}{\chi}+1)^{n-1}=(K+1)\cdot(\frac{K+1}{\epst}+1)^{n-1}=(\nicefrac{\log\thresh}{\varepsilon})^{O(n)}$, and the number of menu entries that cost more than $(1-\epst)^2\cdot\thresh$ in $\mech'$ is the number of menu entries that cost more than $\thresh$ in~$\mech$.
\end{proof}

\end{document}